\documentclass[envcountsame,runningheads]{llncs}
\usepackage{pscproc2}

\usepackage{amssymb}
\usepackage{amsmath}
\usepackage{comment}
\usepackage[disable]{todonotes}
\usepackage{multirow}
\usepackage{listings}
\usepackage[sectionbib, square,sort,comma,numbers]{natbib}
\usepackage{ mathrsfs }
\usepackage{apxproof}

\usepackage{tikz}
\usetikzlibrary{shapes.arrows,chains,positioning}

\makeatletter
\providecommand*{\shuffle}{%
  \mathbin{\mathpalette\shuffle@{}}%
}
\newcommand*{\shuffle@}[2]{%
  \sbox0{$#1\vcenter{}$}%
  \kern .15\ht0 
  \rlap{\vrule height .25\ht0 depth 0pt width 2.5\ht0}%
  \raise.1\ht0\hbox to 2.5\ht0{%
    \vrule height 1.75\ht0 depth -.1\ht0 width .17\ht0 %
    \hfill
    \vrule height 1.75\ht0 depth -.1\ht0 width .17\ht0 %
    \hfill
    \vrule height 1.75\ht0 depth -.1\ht0 width .17\ht0 %
  }%
  \kern .15\ht0 
}
\makeatother

\newcommand{\cshuffle}{\overline{\shuffle}}

\theoremstyle{plain}
\newtheoremrep{theorem}{Theorem}
\newtheoremrep{proposition}[theorem]{Proposition}
\newtheoremrep{lemma}[theorem]{Lemma}
\newtheoremrep{claim}[theorem]{Claim}
\newtheoremrep{conjecture}[theorem]{Conjecture}
\newtheoremrep{corollary}[theorem]{Corollary}
\theoremstyle{definition}
\newtheoremrep{definition}[theorem]{Definition}

\newcommand{\NP}{\textsf{NP}}

\newcommand{\itshuffle}[2]{#1^{\shuffle_{#2},\circledast}}

\title{The \emph{n}-ary Initial Literal and Literal Shuffle}
\titlerunning{The \emph{n}-ary Literal and Initial Literal Shuffle}

\author{Stefan Hoffmann}
\authorrunning{S. Hoffmann}
\institute{Informatikwissenschaften, FB IV, 
  Universit\"at Trier,
  \email{hoffmanns@informatik.uni-trier.de}}

\begin{document}
%

%
%

%
%
\maketitle              
\begin{abstract}
  The literal and the initial literal shuffle have been introduced
  to model the behavior of two synchronized processes.
  However, it is not possible to describe the synchronization of multiple processes.
  Furthermore, both restricted forms of shuffling 
  are not associative.
  Here, we extend the literal shuffle and the initial literal shuffle to multiple arguments.
  We also introduce iterated versions, 
  much different
  from the iterated ones previously introduced for the binary literal and initial literal shuffle.
  We investigate formal properties, and show that in terms of expressive power, in a full trio,
  they coincide with the general shuffle. Furthermore, we look at closure properties 
  with respect to the regular, context-free, context-sensitive, recursive and recursively 
  enumerable languages for all operations introduced.
  Then, we investigate various decision
  problems motivated by analogous problems for the (ordinary) shuffle operation.
  Most problems we look at are tractable, but we also identify one intractable decision problem.
\end{abstract}


\begin{toappendix}

 \paragraph{General Remarks} In this appendix, for a function $f : X \to Y$
 and $Z \subseteq X$, we set $f(Z) = \{ f(z) \mid z \in Z \}$.
\end{toappendix}





\section{Motivation and Contribution}

In~\cite{DBLP:journals/actaC/Berard87,DBLP:journals/tcs/Berard87},
the \emph{initial literal shuffle} and the \emph{literal shuffle}
were introduced, by giving the following natural motivation (taken from~~\cite{DBLP:journals/actaC/Berard87,DBLP:journals/tcs/Berard87}):
\begin{quote} 
\small
[...] The shuffle operation naturally appears in several problems, like concurrency of
processes \cite{Iwana83,Nivat82,Ogden78} or multipoint communication, where all stations share a single
bus~\cite{Iwana83}. That is one of the reasons of the large theoretical literature about this
operation (see, for instance \cite{DBLP:journals/acta/ArakiT81,Ginsburg75,Iwama:1983:UPU,Iwana83,DBLP:journals/tcs/Jantzen81,DBLP:journals/jcss/Latteux79}). In the latter example [of midpoint communication], the general shuffle
operation models the asynchronous case, where each transmitter uses asynchronously
the single communication channel. If the hypothesis of synchronism is
made (step-lock transmission), the situation is modelled by what can be named
‘literal’ shuffle. Each transmitter emits, in turn, one elementary signal. The same
remark holds for concurrency, where general shuffle corresponds to asynchronism
and literal shuffle to synchronism.~[...]
\end{quote}
So, the shuffle operation corresponds to the parallel composition of words, which model instructions or event sequences of processes, i.e., \emph{sequentialized execution histories
of concurrent processes}.

In this framework, the initial literal shuffle is motivated by modelling the synchronous operation
of two processes that start at the same point in time, whereas the literal shuffle
could model the synchronous operation if started at different points in time.
However, both restricted shuffle variants are only binary operations, which are not associative.
Hence, actually only the case of \emph{two} processes synchronized to each other is modelled, i.e., the (initial) literal
shuffling applied multiple times, in any order, does not model adequately the synchronous operation
of multiple processes. 
So, the iterative versions, as introduced in~\cite{DBLP:journals/actaC/Berard87,DBLP:journals/tcs/Berard87},
are not adequate to model multiple processes, and, because of the lack of associativity,
the bracketing is essential, which is, from a mathematical point of view, somewhat unsatisfying.

Here, we built up on~\cite{DBLP:journals/actaC/Berard87,DBLP:journals/tcs/Berard87}
by extending both restricted shuffle variants to multiple arguments, which do not arise
as the combination of binary operations. So, technically, for each $n$
we have a different operation taking $n$ arguments.
With these operations, we derive iterated variants in a uniform manner.
We introduce two variants:
\begin{enumerate}
    \item[(1)] the \emph{$n$-ary initial literal shuffle},
motivated by modelling $n$ synchronous processes started at the same point in time;
 \item[(2)] the \emph{$n$-ary literal shuffle}, motivated by modelling 
$n$ synchronous processes started at different points in time.
\end{enumerate}


Additionally, in Section~\ref{sec:comm_variants}, we also introduce two additional variants
for which the results are independent of the order
of the arguments.
Hence, our variants might be used when a more precise approach is necessary than the general shuffle
can provide.

We study the above mentioned operations and their iterative variants, their relations to each other
and their expressive power. We also study their closure properties with respect
to the classical families of the Chomsky hierarchy~\cite{HopUll79}
and the recursive languages.
We also show that, when adding the
full trio operations, the expressive power of each shuffle variant is as powerful as the general shuffle 
operation. 
In terms of computational complexity, most problem we consider, which are motivated
from related decision problems for the (general) shuffle operation,
are tractable. However, we also identify a decision problem for 
the second variant that is \NP-complete.

The goal of the present work is to give an analysis of these operations
from the point of view of \emph{formal language theory}. 

\section{The Shuffle Operation in Formal Language Theory} 
\label{sec:shuffle_in_formal_lang_theory}


Beside~\cite{DBLP:journals/actaC/Berard87,DBLP:journals/tcs/Berard87}, we briefly review
other work related to the shuffle operation. 
We focus on research in formal language theory and computational complexity.

The shuffle and iterated shuffle have been introduced and studied to understand
the semantics of parallel programs. This was undertaken, as it appears
to be, independently by Campbell and Habermann~\cite{CamHab74}, by Mazurkiewicz~\cite{DBLP:conf/mfcs/Marzurkiewicz75}
and by Shaw~\cite{Shaw78zbMATH03592960}. They introduced \emph{flow expressions}, 
which allow for sequential operators (catenation and iterated catenation) as well
as for parallel operators (shuffle and iterated shuffle). 
See also~\cite{DBLP:journals/cl/Riddle79,DBLP:journals/cl/Riddle79a}
for an approach using only the ordinary shuffle, but not the iterated shuffle.
Starting from this, various subclasses of the flow expressions were
investigated~\cite{DBLP:journals/tcs/BerglundBB13,DBLP:journals/tcs/Jantzen81,DBLP:journals/tcs/Jantzen85,DBLP:journals/ipl/Jedrzejowicz83,DBLP:journals/tcs/JedrzejowiczS01,DBLP:journals/jcss/WarmuthH84,DBLP:journals/jcss/BussS14,DBLP:journals/fuin/KudlekF14}.

Beside the literal shuffle~\cite{DBLP:journals/actaC/Berard87,DBLP:journals/tcs/Berard87}, and the variants introduced in this work,
other variants and generalizations of the shuffle product were introduced.
Maybe the most versatile is \emph{shuffle by trajectories}~\cite{DBLP:journals/tcs/MateescuRS98},
whereby the selection of the letters from two input words is controlled by a given
\emph{language of trajectories} that indicates, as a binary language, at which positions
letters from the first or second word are allowed. This framework entails numerous
existing operations, from concatenation up to the general shuffle, and
in~\cite{DBLP:journals/tcs/MateescuRS98} the authors related algebraic properties
and decision procedures of resulting operations to properties of the trajectory language.
In~\cite{DBLP:journals/tcs/BeekMM05}, with a similar motivation as ours, different notions
of \emph{synchronized shuffles} were introduced.
But in this approach, two words have to ``link'', or synchronize, at a specified subword drawn
from a subalphabet (the letters, or actions, that should be synchronized),
which the authors termed the \emph{backbone}. Hence, their approach differs in that
the synchronization appears letter-wise, whereas here we synchronize position-wise, i.e., at specific points in time the actions occur together in steps, and are not merged as in~\cite{DBLP:journals/tcs/BeekMM05}.


%
 
\section{Preliminaries and Definitions}
\label{sec::preliminaries}

%
%


By $\mathbb N_0$ we denote the \emph{natural numbers} including zero.
The \emph{symmetric group}, i.e., the set
of all permutations with function composition as operation,
is $\mathcal S_n = \{ f: \{1,\ldots, n\} \to \{1,\ldots, n\} \mid \mbox{$f$ bijective} \}$.

By $\Sigma$ we  denote a finite set of symbols, called an \emph{alphabet}. The set $\Sigma^{\ast}$ denotes
the set of all finite sequences, i.e., of all words with
the concatenation operation. The finite sequence of length zero,
or the \emph{empty word}, is denoted by~$\varepsilon$. 
Subsets of $\Sigma^{\ast}$
are called \emph{languages}. 
For a given word, we denote by $|w|$
its length, and for $a \in \Sigma$ by $|w|_a$ the number of occurrences of the symbol $a$
in $w$. 
For a word $w = u_1 \cdots u_n$ with $u_i \in \Sigma$, $i \in \{1,\ldots,n\}$,
we write $w^R = u_n \cdots u_1$ for the \emph{mirror operation}.
For $L \subseteq \Sigma^*$ we set $L^+ = \bigcup_{i=1}^{\infty} L^i$
and $L^* = L^+ \cup \{\varepsilon\}$, where we set $L^1 = L$
and $L^{i+1} = \{ uv \mid u \in L^i, v \in L \}$ for $i \ge 1$.

A \emph{finite deterministic and complete automaton} will
be denoted by $\mathcal A = (\Sigma, S, \delta, s_0, F)$
with $\delta : S \times \Sigma \to S$ the state transition function, $S$ a finite set of states, $s_0 \in S$
the start state and $F \subseteq S$ the set of final states. 
The properties of being deterministic and complete are implied by the definition of $\delta$
as a total function.
The transition function $\delta : S \times \Sigma \to S$
could be extended to a transition function on words $\delta^{\ast} : S \times \Sigma^{\ast} \to S$
by setting $\delta^{\ast}(s, \varepsilon) := s$ and $\delta^{\ast}(s, wa) := \delta(\delta^{\ast}(s, w), a)$
for $s \in S$, $a \in \Sigma$ and $w \in \Sigma^{\ast}$. In the remainder we drop
the distinction between both functions and also denote this extension by $\delta$.
The \emph{language accepted} by an automaton $\mathcal A = (\Sigma, S, \delta, s_0, F)$ is
$
 L(\mathcal A) = \{ w \in \Sigma^{\ast} \mid \delta(s_0, w) \in F \}.
$
A language $L \subseteq \Sigma^{\ast}$ is called \emph{regular} if $L = L(\mathcal A)$
for some finite automaton. 


\begin{definition} The \emph{shuffle operation}, denoted by $\shuffle$, is defined by
\label{def:shuffle}
 \begin{multline*}
    u \shuffle v  = \{ w \in \Sigma^*  \mid  w = x_1 y_1 x_2 y_2 \cdots x_n y_n 
    \emph{ for some words } \\ x_1, \ldots, x_n, y_1, \ldots, y_n \in \Sigma^*
    \emph{ such that } u = x_1 x_2 \cdots x_n \emph{ and } v = y_1 y_2 \cdots y_n \},
 \end{multline*}
 for $u,v \in \Sigma^{\ast}$ and 
  $L_1 \shuffle L_2  := \bigcup_{x \in L_1, y \in L_2} (x \shuffle y)$ for $L_1, L_2 \subseteq \Sigma^{\ast}$.
\end{definition}

 \begin{example} 
  $\{ab\} \shuffle \{cd\} = \{ abcd, acbd, acdb, cadb, cdab, cabd \}$
 \end{example}

 The shuffle operation is commutative, associative and distributive over union. We will use these properties without further mention. In writing formulas
 without brackets we suppose that the shuffle operation binds stronger than the set operations,
 and the concatenation operator has the strongest binding. 
 For $L \subseteq \Sigma^*$ the \emph{iterated shuffle} is $L^{\shuffle,*} = \bigcup_{i=0}^{\infty} L^{\shuffle, i}$
 with $L^{\shuffle, 0} = \{\varepsilon\}$ and $L^{\shuffle, i + 1} = L \shuffle L^{\shuffle, i}$.
 The \emph{positive iterated shuffle} is $L^{\shuffle,+} = \bigcup_{i=1}^{\infty} L^{\shuffle, i}$.
 
\begin{toappendix}
For reference, we state the next relation between the shuffle operation
and homomorphic mappings.

\begin{lemmarep}\label{lem:shuffle_hom_inclusion}
 Let $h : \Sigma^* \to \Gamma^*$ be a homomorphism for the concanentation
 operation\footnote{So, I do not assume it is a homomorphism for the shuffle operation.}
 and $U, V \subseteq \Sigma^*$.
 Then
 $
  h(U \shuffle V) \subseteq h(U) \shuffle h(V).
 $
\end{lemmarep}
\begin{proof}
 Let $w \in U \shuffle V$.
 Then, by Definition~\ref{def:shuffle},
 $w = u_1 v_1 u_2 v_2 \cdots u_n v_n$
 for some $n \ge 1$, $u_1 u_2 \cdots u_n = u \in U$
 and $v_1 v_2 \cdots v_n = v \in V$.
 So, $$
 h(w) = h(u_1) h(v_1) h(u_2) h(v_2) \cdots h(u_n) h(v_n)
 $$
 and $h(u_1 u_2 \cdots u_n) \in h(U)$, $h(v_1 v_2 \cdots v_n) \in h(V)$.
 Hence, again using Definition~\ref{def:shuffle},
 we find $h(w) \in h(U)\shuffle h(V)$.~\qed
\end{proof}
\end{toappendix}

A \emph{full trio}~\cite{GinsburgGreibach67} is a family of languages closed under homomorphisms, inverse homomorphisms 
and intersections with regular sets. A full trio is closed under arbitrary intersection
if and only if it is closed under shuffle~\cite{Ginsburg75}. Also, by a theorem of Nivat~\cite{Nivat68}, 
a family of languages forms a full trio if and only if it is closed under
\emph{generalized sequential machine mappings (gsm mappings)}, also called \emph{finite state transductions}. 
For the definition of gsm mappings, as well as of \emph{context-sensitive} and \emph{recursively enumerable languages},
we refer to the literature, for example~\cite{HopUll79}. 
For two arguments, the \emph{interleaving operation} (or \emph{perfect shuffle}~\cite{DBLP:journals/eatcs/HenshallRS12})
was introduced
in~\cite{DBLP:journals/actaC/Berard87,DBLP:journals/tcs/Berard87}.
Here, we give a straightforward generalization for multiple, equal-length,
input words. 

\begin{definition}[$n$-ary interleaving operation]
\label{def:interleaving_operator}
 Let $n \ge 1$, $k \ge 0$, $u_1, \ldots, u_n \in \Sigma^k$.
 If $k > 0$, write  $u_i = x_1^{(i)} \cdots x_k^{(i)}$, $x_j^{(i)} \in \Sigma$ for $j \in \{1,\ldots,k\}$
 and $i \in \{1,\ldots, n\}$.
 Then we define $I : (\Sigma^k)^n \to \Sigma^{nk}$
 by
 $$
 I(u_1, \ldots, u_n) = x_1^{(1)} \cdots x_1^{(n)} x_2^{(1)} \cdots x_2^{(n)} \cdot \ldots \cdot x_k^{(1)} \cdots x_k^{(n)}.
 $$
 If $k = 0$, then $I(\varepsilon, \ldots, \varepsilon) = \varepsilon$.
\end{definition}

\begin{example} $I(aab, bbb, aaa) = abaababba$. \end{example}

If we interleave all equal-length words in given regular languages,
the resulting language is still regular.

\begin{propositionrep}
 Let $L_1, \ldots, L_n \subseteq \Sigma^*$ be regular.
 Then 
 $ 
  \{ I(u_1, \ldots, u_n) \mid \exists m \ge 0 \ \forall i \in \{1,\ldots,n\} : u_i \in L_i \cap \Sigma^m \}
 $
 is regular.
\end{propositionrep}
\begin{proof}
 Let $\mathscr A_i = (\Sigma, Q_i, \delta_i, s_i, F_i)$
 be accepting automata for $L_i$ with $i \in \{1,\ldots, n\}$.
 Construct $\mathscr A = (\Sigma, Q, \delta, s_0, F)$
 with $Q = Q_1 \times \ldots \times Q_n \times \{0, \ldots, n-1\}$
 and
 \begin{align*}
  \delta((q_1, \ldots, q_n, k), x) & = (q_1, \ldots, q_{k-1}, \delta_k(q_k, x), q_{k+1}, \ldots, q_n), (k+1) \bmod n) \\
   s_0 & = (s_1, \ldots, s_n, 0) \\
   F   & = F_1 \times \ldots \times F_n \times \{0\}.
 \end{align*}
 Then it is easy to see that $L(\mathscr A) = \{ I(u_1, \ldots, u_n) \mid \exists m \forall i \in \{1,\ldots,n\} : u_i \in L_i \cap \Sigma^m \}$.~\qed
\end{proof}

In~\cite{DBLP:journals/actaC/Berard87,DBLP:journals/tcs/Berard87},
the \emph{initial literal shuffle} and the \emph{literal shuffle}
were introduced.

\begin{definition}[\cite{DBLP:journals/tcs/Berard87,DBLP:journals/actaC/Berard87}]
 Let $U,V \subseteq \Sigma^*$.
 The \emph{initial literal shuffle} of $U$ and $V$
 is
 \[ 
  U \shuffle_1 V = \{ I(u,v)w \mid u,v,w \in \Sigma^*, |u| = |v|,
                     (uw \in U, v \in V) \mbox{ or } (u \in U, vw \in V) \}.
 \] 
 and the \emph{literal shuffle} is
 \begin{align*}  
  U \shuffle_2 V = \{ & w_1 I(u,v) w_2 \mid w_1, u,v, w_2 \in \Sigma^*, |u| = |v|, \\
   & (w_1 u w_2 \in U, v \in V) \mbox{ or } 
    (u \in U, w_1 v w_2 \in V) \mbox{ or } \\
   & (w_1 u \in U, vw_2 \in V) \mbox{ or } 
    (uw_2 \in U, w_1 v \in V) \}.
 \end{align*}
\end{definition}

\begin{example}
 $\{ abc \} \shuffle_1 \{ de \} = \{ adbec \}$,
 $\{ abc \} \shuffle_2 \{ de \} = \{ abcde, abdce, adbec, daebc, deabc \}$.
\end{example}

The following iterative variants were introduced in~\cite{DBLP:journals/actaC/Berard87,DBLP:journals/tcs/Berard87}.

\begin{definition}[\cite{DBLP:journals/tcs/Berard87,DBLP:journals/actaC/Berard87}] 
 Let $L \subseteq \Sigma^*$. For $i \in \{1,2\}$, set
 $$
   L^{\shuffle_i^*} = \bigcup_{n \ge 0} L_n, \mbox{ where } L_0 = \{\varepsilon\} \mbox{ and } L_{n+1} = L_n \shuffle_i L.
 $$
\end{definition}

The next results are stated in~\cite{DBLP:journals/actaC/Berard87,DBLP:journals/tcs/Berard87}.

\begin{proposition}[\cite{DBLP:journals/actaC/Berard87,DBLP:journals/tcs/Berard87}] 
\label{prop:trio_shuffle_lit_init} 
 Let $\mathcal L$ be a full trio. The following are equivalent:
 \begin{enumerate} 
 \item $\mathcal L$ is closed under shuffle.
 \item $\mathcal L$ is closed under literal shuffle.
 \item $\mathcal L$ is closed under initial literal shuffle.
 \end{enumerate}
\end{proposition}

\begin{proposition}[\cite{DBLP:journals/actaC/Berard87,DBLP:journals/tcs/Berard87}]
\label{prop:finite_binary_iterative}
 Let $F \subseteq \Sigma^*$ be finite. Then $F^{\shuffle_1^*}$ is regular.
\end{proposition}




\section{The \emph{n}-ary Initial Literal and Literal Shuffle}
\label{sec:naryvariants}

%


Here, for any number of arguments, 
we introduce both shuffle variants, define
iterated versions and state basic properties.

\begin{definition}
\label{lem:hom_char}
 Let $u_1, \ldots, u_n \in \Sigma^*$
 and $N = \max\{ |u_i| \mid i \in \{1,\ldots,n\} \}$.
 Set
 \begin{enumerate}
 \item 
 $\shuffle_1^n(u_1, \ldots, u_n) 
   = h(I(u_1 \$^{N - |u_1|}, \ldots, u_n \$^{N - |u_n|}))$ and 
 \item $
  \shuffle_2^n(u_1, \ldots, u_n) = \{ h(I(v_1, \ldots, v_n)) \mid v_i \in U_i, i \in \{1,\ldots,n\} \}$, 
 \end{enumerate}
  where $U_i = \{ \$^k u_i \$^{r - k} \mid 0 \le k \le r \mbox{ with } r = n\cdot N - |u_i| \}$
 for $i \in \{1,\ldots,n\}$
 and $h : (\Sigma \cup \{\$\})^* \to \Sigma^*$ is
 the homomorphism given by $h(\$) = \varepsilon$
 and $h(x) = x$ for $x \in \Sigma$.
 
\end{definition}

Note that writing the number of arguments in the upper index should pose no problem
or confusion with the power operator on functions, as the $n$-ary shuffle
variants are only of interest for $n \ge 2$, and raising a function to a power only makes
sense for function with a single argument.

For languages $L_1, \ldots, L_n \subseteq \Sigma^*$, we set 
\begin{align*} 
 \shuffle_1^n(L_1,\ldots, L_n) & = \bigcup_{u_1 \in L_1, \ldots u_n \in L_n} \{ \shuffle_1^n(u_1, \ldots, u_n) \} \\ 
 \shuffle_2^n(L_1,\ldots, L_n) & = \bigcup_{u_1 \in L_1, \ldots u_n \in L_n}\shuffle_2^n(u_1, \ldots, u_n).
\end{align*}

\begin{example} 
 \label{ex:shuffle_2_abbc}
 Let $u = a, v = bb, w = c$. Then
 $\shuffle_1^3(u,v,w) = abcb$
 and
  \begin{align*} 
   \shuffle_2^3(u,v,w)
    & = \{ bbac, babc, abcb, acbb, abbc,
         bbca, bcba, bcab, cbab, cabb, cbba \} \\
   \shuffle_2^3(v,u,w) & = 
    \{ bbac, babc, bacb, abcb, abbc, acbb, cbab, cbba, cabb, bcba, bbca \}
  \end{align*}
  We see $bacb \notin \shuffle_2^3(u,v,w)$, but $bacb \in \shuffle_2^3(v,u,w)$.
\end{example}

\begin{example} 
\label{ex:graphical_depiction}
Please see Figure~\ref{fig:ex_shuffle2}
for a graphical depiction of the word
\[ a_1^{(1)} a_2^{(1)} a_3^{(1)} a_4^{(1)} a_1^{(2)} a_5^{(1)} a_2^{(2)} a_6^{(1)} 
 	a_3^{(2)} a_1^{(3)} a_7^{(1)} a_4^{(2)} a_2^{(3)} a_5^{(2)} a_3^{(3)}
 	a_6^{(2)} a_7^{(2)} a_8^{(2)} a_9^{(2)} 
\]
  from
$\shuffle_2^3(u, v, w)$
with \begin{align*}
 u & = a_1^{(1)} a_2^{(1)} a_3^{(1)} a_4^{(1)} a_5^{(1)} a_6^{(1)} a_7^{(1)}, \\ 
 v & = a_1^{(2)} a_2^{(2)} a_3^{(2)} a_4^{(2)} a_5^{(2)} a_6^{(2)} a_7^{(2)} a_8^{(2)} a_9^{(2)}, \\
 w & = a_1^{(3)} a_2^{(3)} a_3^{(3)}.
 \end{align*}
\begin{figure}
	\centering
\scalebox{1.0}{
 \begin{tikzpicture}[
      start chain=1 going right,start chain=2 going right,node distance=-0.15mm
    ]
    \foreach \x in {1,2,...,7} {
        \node [draw,on chain=1] (first\x) {$a_{\x}^{(1)}$};
    } 
    
    \node[draw,on chain=1,below=of first4] {$a_1^{(2)}$};
    \foreach \x in {2,...,9} {
        \node[draw,on chain=1] (second\x) {$a_{\x}^{(2)}$};
    }    
    
    \node[draw,on chain=1,below=of second3] {$a_1^{(3)}$};
    \foreach \x in {2,...,3} {
        \node[draw,on chain=1] {$a_{\x}^{(3)}$};
    }    
    
\end{tikzpicture}
}
	\caption{
	Graphical depiction of a
 	word in $\shuffle_2^3(u, v, w)$.
 	See Example~\ref{ex:graphical_depiction}.
}
\label{fig:ex_shuffle2}
\end{figure}
\end{example}

\begin{toappendix}
In various proofs in this section, we will
need the next lemma.

\begin{lemma} 
\label{lem:shuffle_2_form}
  Let $u_1, \ldots, u_n \in \Sigma$, $\$\notin \Sigma$
 and $h : (\Sigma \cup \{\$\})^* \to \Sigma^*$ be
 the homomorphism given by $h(x) = x$ for $x \in \Sigma$
 and $h(\$) = \varepsilon$. 
 Let $N > \sum_{i=1}^n |u_i|$. Then,
 \begin{enumerate}
 \item 
 \[
  \shuffle_2^n(u_1, \ldots, u_n) = h(I(\$^* u_1 \$^* \cap \Sigma^N, \ldots, \$^* u_n \$^* \cap \Sigma^N)),
 \]

 \item 
 \[
  \shuffle_2^n(u_1, \ldots, u_n) = \bigcup_{0 \le k_1, \ldots, k_n \le N} \{ h(\shuffle_1^n(\$^{k_1} u_1, \ldots, \$^{k_n} u_n)) \}.
 \]
 
 \item 
 $
  \shuffle_2^n(u_1, \ldots, u_n)
   = h(\{ I(v_1, \ldots, v_n) \mid \exists m \forall i \in \{1,\ldots, n\} : v_i \in \$^* u_i \$^* \cap \Sigma^m \}).
 $
\end{enumerate}
\end{lemma}
\begin{proof} 
 We only show that last equation, the others could be shown similarly.
 First, observe that,
 \begin{align*} 
  & \{ I(v_1, \ldots, v_n) \mid \exists m \forall i \in \{1,\ldots, n\} : v_i \in \$^* u_i \$^* \cap \Sigma^m \} \\
  & = \{ I(v_1, \ldots, v_n) \mid \exists m \forall i \in \{1,\ldots, n\} \exists k: v_i = \$^{k} u_i \$^{m - |u_i| - k}, m \ge |u_i| + k \}.
 \end{align*}
 We have
 $$
  \shuffle_2^n(u_1, \ldots, u_n) = \bigcup_{k_1, \ldots, k_n \in \mathbb N_0} \{ h(\shuffle_1(\$^{k_1} u_1, \ldots, \$^{k_n} u_n)) \}.
 $$
 By Definition~\ref{lem:hom_char}, but using the additional sign $\&$
 and the homomorphism $g : \Sigma \cup \{\$, \& \} \to \Sigma \cup \{\$\}$,
 we find
 $$
  \shuffle_1(\$^{k_1} u_1, \ldots, \$^{k_n} u_n)
   = g(I(\$^{k_1} u_1 \&^{m - \max_i |u_i| + k_i}, \ldots, \$^{k_n} u_n \&^{m - \max_i |u_i| + k_i}))
 $$
 for each $m \ge \max_i |u_i| + k_i$. We can write
 \begin{align*} 
  & h(g(I(\$^{k_1} u_1 \&^{m - \max_i |u_i| + k_i}, \ldots, \$^{k_n} u_n \&^{m - \max_i |u_i| + k_i}))) \\
   & = h(I(\$^{k_1} u_1 \$^{m - \max_i |u_i| + k_i}, \ldots, \$^{k_n} u_n \$^{m - \max_i |u_i| + k_i})).
 \end{align*}
 For each $m \ge \max_i |u_i| + k_i$, setting $v_i = \$^{k_i} u_i \$^{m - \max_i |u_i| + k_i}$
 we find
 \begin{align*} 
  & I(\$^{k_1} u_1 \$^{m - \max_i |u_i| + k_i}, \ldots, \$^{k_n} u_n \$^{m - \max_i |u_i| + k_i}) \\
  & \subseteq \{ I(v_1, \ldots, v_n) \mid \exists m \forall i \in \{1,\ldots, n\} \exists k: v_i = \$^{k} u_i \$^{m - |u_i| - k}, m \ge |u_i| + k \}
 \end{align*}
 So,
 \begin{align*}
  & \bigcup_{k_1, \ldots, k_n \in \mathbb N_0} \{ h(\shuffle_1(\$^{k_1} u_1, \ldots, \$^{k_n} u_n)) \} \\
    & \subseteq h(\{ I(v_1, \ldots, v_n) \mid \exists m \forall i \in \{1,\ldots, n\} \exists k: v_i = \$^{k} u_i \$^{m - |u_i| - k}, m \ge |u_i| + k \})
 \end{align*}
 The other inclusion is clear.~\qed 
\end{proof}
\end{toappendix}

Now, we can show that the two introduced $n$-ary literal shuffle variants
generalize the initial literal shuffle and the literal shuffle from~\cite{DBLP:journals/actaC/Berard87,DBLP:journals/tcs/Berard87}.

\begin{lemmarep}
\label{lem:n_2_equals_init_and_literal_shuffle}
 Let $U, V \subseteq \Sigma^*$ be two languages. Then
 $$
  \shuffle_1(U,V) = \shuffle_1^2(U,V)
 \quad\mbox{and}\quad
  \shuffle_2(U,V) = \shuffle_2^2(U,V).
 $$
\end{lemmarep}
\begin{proof}
  Let $w \in \shuffle_1^2(U,V)$.
  Then $w = \shuffle_1^2(u,v)$ for $u \in U$, $v \in V$.
  Without loss of generality, suppose $|u| \le |v|$.
  By Definition~\ref{lem:hom_char}, we have $w = h(I(u\$^{|v| - |u|}, v))$.
  Write $v = xy$ with $|x| = |u|$
  So, 
  $$
   w = h(I(u, x) I(\$^{|v| - |u|}, y))
     = h(I(u,x)) h(I(\$^{|v| - |u|}, y)))
     = I(u,x) y.
  $$
  Hence $w \in \shuffle_1(U,V)$.
  Conversely, suppose $w \in \shuffle_1(U,V)$.
  Then $w = I(x,v)y$ with $xy \in U$, $v \in V$, $|x| = |v|$
  or $w = I(u, x)y$ with $u \in U$, $xy \in V$, $|x| = |u|$.
  In the first case, $w = h(I(xy, u\$^{|y|}))$.
  So $w \in \shuffle_1^2(U,V)$ by Definition~\ref{lem:hom_char}.
  Similarly in the other case.

  The other equality could be shown similarly to the first.~\qed
\end{proof}

We can also write the second $n$-ary variant in terms of the first and the mirror operation, as stated in the next lemma.

\begin{lemmarep}\label{lem:reduce_2_to_1}
 Let $u_1, \ldots, u_n \in \Sigma^*$.
 Then
 \[ 
   \shuffle_2^n(u_1, \ldots, u_n) = \bigcup_{ \substack{ x_1, \ldots, x_n \in \Sigma^* \\ y_1, \ldots, y_n \in \Sigma^* \\ u_i = x_i y_i }} \{ ( \shuffle_1^n(x_1^R, \ldots, x_n^R) )^R \cdot \shuffle_1^n(y_1, \ldots, y_n) \}
 \]
\end{lemmarep}
\begin{proof}
 Note that, for any $j \in \{1,\ldots, n\}$, we can write
 \begin{equation}
 \label{eqn:include_dollar_at_end}
  \shuffle_n^1(u_1, \ldots, u_n) = h(\shuffle_n^1(u_1, \ldots, u_{j-1}, u_j\$, u_{j+1},\ldots, u_n)).
 \end{equation} 
 We have
 \[ 
  \shuffle_n^2(u_1, \ldots, u_n)
   = \bigcup_{0 \le k_1, \ldots, k_n \le m} \{ h(\shuffle_n^1(\$^{k_1} u_1, \ldots, \$^{k_n} u_n)) \}.
 \]
 First, let $u = h(\shuffle_n^1(\$^{k_1} u_1, \ldots, \$^{k_n} u_n))$
 for some $0 \le k_1, \ldots, k_n \le m$.
 Without loss of generality, suppose $k_1 \ge k_2 \ge \ldots \ge k_n$.
 With the above Equation~\eqref{eqn:include_dollar_at_end}, we can write
 \[
  u = h(\shuffle_n^1(\$^{k_1} u_1 \$^m, \ldots, \$^{k_n} u_n \$^m)).
 \]
 Now, for each $i \in \{1,\ldots, n\}$, write
 \[
  \$^{k_i} u_i \$^m = \$^{k_i} x_i y_i \$^m
 \]
 with $k_i + |x_i| = k_1$ and $x_i, y_i \in \Sigma^*$. Then 
 \[
  \shuffle_n^1(\$^{k_1} u_1 \$^m, \ldots, \$^{k_n} u_n \$^m)
   = I( \$^{k_1} x_1, \ldots, \$^{k_n} x_n ) \cdot \shuffle_n^1( y_1 \$^m, \ldots, y_n \$^m ).
 \]
 Furthermore 
 \[
  I( \$^{k_1} x_1, \ldots, \$^{k_n} x_n )
   = I( x_1^R \$^{k_1}, \ldots, x_n^R \$^{k_n})^R.
 \]
 Hence
 \begin{align*} 
     u &  = h(\shuffle_n^1(\$^{k_1} u_1 \$^m, \ldots, \$^{k_n} u_n \$^m)) \\
       &  = h( I( x_1^R \$^{k_1}, \ldots, x_n^R \$^{k_n})^R \cdot \shuffle_n^1( y_1 \$^m, \ldots, y_n \$^m )) \\
       &  = h( I( x_1^R \$^{k_1}, \ldots, x_n^R \$^{k_n})^R ) \cdot h( \shuffle_n^1( y_1 \$^m, \ldots, y_n \$^m ) ) \\
       &  = h( \shuffle_n^1( x_1^R \$^{k_1}, \ldots, x_n^R \$^{k_n})^R ) \cdot h( \shuffle_n^1( y_1 \$^m, \ldots, y_n \$^m ) ) \\
       &  = h( \shuffle_n^1( x_1^R \$^{k_1}, \ldots, x_n^R \$^{k_n}) )^R \cdot h( \shuffle_n^1( y_1 \$^m, \ldots, y_n \$^m ) ) \\
       &  = \shuffle_n^1( x_1^R, \ldots, x_n^R )^R \cdot \shuffle_n^1( y_1, \ldots, y_n).
 \end{align*}

 Conversely, let $u = \shuffle_n^1( x_1^R, \ldots, x_n^R )^R \cdot \shuffle_n^1( y_1, \ldots, y_n)$
 with $u_i = x_i y_i$ for $\in \{1,\ldots, n\}$.
 Then, choosing numbers $0 \le k_1, \ldots, k_n \le m$
 such that $|x_i + k_i| = \max\{ |x_j| \mid j \in \{1,\ldots,n\}\}$ we can read
 the above equation in the reverse order and find
 \[
  u = h(\shuffle_n^1(\$^{k_1} u_1 \$^m, \ldots, \$^{k_n} u_n \$^m)) 
    = h(\shuffle_n^1(\$^{k_1} u_1, \ldots, \$^{k_n} u_n)) \in \shuffle_n^2(u_1, \ldots, u_n).
 \] 
\end{proof}

With these multiple-argument versions, we define an ``iterated'' version, where iteration
is not meant in the usual sense because of the lack of associativity
for the binary argument variants.

\begin{definition} 
\label{def:iterated_versions}
Let $L \subseteq \Sigma^*$
be a language. Then,  for $i \in \{1,2\}$, define
 \begin{align*}
   \itshuffle{L}{i}  = \{ \varepsilon \} \cup \bigcup_{n \ge 1} \shuffle_i^n(L, \ldots, L).
 \end{align*}

\end{definition}

For any $i \in \{1,2\}$, as $\shuffle_i^1(L) = L$,
we have $L \subseteq \itshuffle{L}{i}$.
%
%
%
%
%
Now, let us investigate some properties of the operations under consideration.

\begin{propositionrep}
\label{prop:properties_shuffles}
 Let $L_1, \ldots, L_n \subseteq \Sigma^*$ and $\pi : \{1,\ldots,n\} \to \{1,\ldots,n\}$ a permutation. Then 
 \begin{enumerate}
 \item \label{prop:properties_shuffles:concat_in_shuffle_2}
  $
   L_{\pi(1)} \cdots L_{\pi(n)} \subseteq \shuffle_2^n(L_1, \ldots, L_n).
  $
 \item \label{prop:properties_shuffles:cyclic_permutation}
 Let $k \in \mathbb N_0$. Then $
  \shuffle_2^n(L_1, \ldots, L_n) = \shuffle_2^n(L_{((1 + k-1)\bmod n) + 1}, \ldots ,L_{((n + k-1)\bmod n) + 1}).
 $
 \item \label{prop:properties_shuffles:shuffle_1_in_shuffle_2}
  $\shuffle_1^n(L_1, \ldots, L_n) \subseteq \shuffle_2^n(L_1, \ldots, L_n) \subseteq L_1 \shuffle \cdots \shuffle L_n$;
 \item \label{prop:properties_shuffles:Kleene_in_it_shuffle_2} $L_1^* \subseteq \itshuffle{L_1}{2}$;
 \item $\Sigma^* = \itshuffle{\Sigma}{i}$ for $i \in \{1,2\}$; 
 \item $\itshuffle{L_1}{1} \subseteq \itshuffle{L_1}{2} \subseteq L_1^{\shuffle,*}$;
 \item for $u_1, \ldots, u_n, u \in \Sigma^*$, 
  if $u \in \shuffle_i^n(\{u_1\}, \ldots, \{u_n\})$, then $|u| = |u_1| + \cdots + |u_n|$.
 \end{enumerate}
\end{propositionrep}
\begin{proofsketch}
 We only give a rough outline for Property~\ref{prop:properties_shuffles:cyclic_permutation}.
 Let $x_{i,j} \in \Sigma$ for $i \in \{1,\ldots,n\}$, $j \in \{1,\ldots,m\}$.
 Then, the main idea is to use the equations
 \begin{align*}
    & h( I(x_{1,1} \cdots x_{1,m},  x_{2,1} \cdots x_{2,m},  \ldots, x_{n,1} \cdots x_{n,m} ) ) \\
    & = h( ( x_{1,1} \cdots x_{n,1} ) ( x_{1,2} \cdots x_{n,2} )  \cdots  ( x_{1,m} \cdots x_{n,m} ) ) \\
    & = h( \$^n ( x_{1,1} \cdots x_{n,1} ) ( x_{1,2} \cdots x_{n,2} )  \cdots ( x_{1,m} \cdots x_{n,m} ) ) \\
    & = h( ( \$^{n-1} x_{1,1} ) ( x_{2,1} \cdots x_{n,1} x_{1,2} )  \cdots  ( x_{2,m} \cdots x_{n,m} \$ )) \\ 
    & = h( I( \$ x_{2,1} \cdots x_{2,m}, \$ x_{3,1} \cdots x_{3,m}, \ldots, \$ x_{n,1} \cdots x_{n,m}, x_{1,1} \cdots x_{1,m} \$ ))
 \end{align*}
 and
  $
  \shuffle_2^n(u_1, \ldots, u_n)
   = h(\{ I(v_1, \ldots, v_n) \mid \exists m \forall i \in \{1,\ldots, n\} : v_i \in \$^* u_i \$^* \cap \Sigma^m \})
 $
 with $h : (\Sigma \cup \{\$\})^* \to \Sigma^*$ as in Definition~\ref{lem:hom_char}.~\qed
\end{proofsketch}
\begin{proof} 
 We will use Lemma~\ref{lem:shuffle_2_form}
 at various places without special mentioning.

\begin{enumerate}
\item 
 For $u_1, \ldots, u_n \in \Sigma^*$, we have
 \begin{align*} 
  u_1 u_2 \cdots u_n = h(\shuffle_1^n(u_1, \$^{|u_1|} u_2, \$^{|u_1|+|u_2|} u_3, \ldots, \$^{|u_1|+\ldots + |u_{n-1}|} u_n).
 \end{align*}
 And, in general, for  $\pi : \{1,\ldots,n\} \to \{1,\ldots,n\}$,
 setting\footnote{With the definition $\sum_{j=1}^0 a_i = 0$.} 
 \[ 
  v_i = \$^{\sum_{j=1}^{\pi^{-1}(i)-1} |u_{\pi^{-1}(j)}|} 
 \]
 for $i \in \{1,\ldots,n\}$,
 we find
 \[ 
  u_{\pi(1)} u_{\pi(2)} \cdots u_{\pi(n)} 
   =  h(\shuffle_1^n(v_1 u_1, \ldots, v_n u_n)).
 \]
\item Let $u \in \shuffle_2^n(u_1, \ldots, u_n)$ with $\{ u_1, \ldots, u_n \}\subseteq L$.
 By Definition~\ref{lem:hom_char}, we have
 \[
  u = h(I(v_1, \ldots, v_n))
 \]
 with $m = |v_1| = \ldots = |v_n|$ and 
 $v_i = \$^{k_i} u_i \$^{r_i - k_i}$ with $r_i = m - |u_i|$ and $0 \le k_i \le r_i$.
 Write $v_i = x_{i,1} \cdots x_{i,m}$ with $x_{i,j} \in \Sigma \cup \{\$\}$
 for $j \in \{1,\ldots,m\}$.
 Then,
 \[
  u = h( ( x_{1,1} \cdots x_{n,1} ) ( x_{1,2} \cdots x_{n,2} ) \cdot \ldots \cdot ( x_{1,m} \cdots x_{n,m} ) ), 
 \]
 i.e., reading the matrix $( x_{i,j} )_{1 \le i \le n, 1 \le j \le m}$
 column-wise and concatenating the result.
 We have
 \begin{align*}
       u & = h( \$^n ( x_{1,1} \cdots x_{n,1} ) ( x_{1,2} \cdots x_{n,2} ) \cdot \ldots \cdot ( x_{1,m} \cdots x_{n,m} ) ) \\
       & = h( ( \$^{n-1} x_{1,1} ) ( x_{2,1} \cdots x_{n,1} x_{1,2} ) \cdot \ldots \cdot  ( x_{2,m} \cdots x_{n,m} \$ )). 
 \end{align*}
 The second line results, by shifting the parenthesis one step to the right, removing the initial dollar sign,
 and appending one dollar sign to the end.
 The way the parenthesis are set could be interpreted as reading a different matrix, namely
 \[ 
  \begin{pmatrix}
   \$ & x_{2,1} & x_{2,2} & \cdots & x_{2,m-1} & x_{2,m} \\
   \$ & x_{3,1} & x_{3,2} & \cdots & x_{3,m-1} & x_{3,m} \\
       \vdots & & & \ddots & & \vdots \\ 
   \$ & x_{n,1} & x_{n,2} & \cdots & x_{n,m-1} & x_{n,m} \\
   x_{1,1} & x_{1,2} & x_{1,3} & \cdots & x_{1,m} & \$
  \end{pmatrix}
 \]
 column-wise and appending the result.
 So, we can write
 \begin{align*}
  u & = h( ( \$^{n-1} x_{1,1} ) ( x_{2,1} \cdots x_{n,1} x_{1,2} ) \cdot \ldots \cdot  ( x_{2,m} \cdots x_{n,m} \$ )) \\
    & = h(I( \$ v_2, \$ v_3, \ldots, \$ v_n, v_1 \$ )).
 \end{align*}
 Set $w_i = \$ v_{i+1} = \$^{k_i + 1} u_{i+1} \$^{r_i - k_i}$
 for $i \in \{1,\ldots,n-1\}$, $w_{n} = v_1 \$ = \$^{k_1} u_1 \$^{r_1 - k_1 + 1}$.
 We have, as $|w_1| = \ldots = |w_n| = m + 1$,
 \[
  w_i = \left\{ 
  \begin{array}{ll}
   \$^{k_i + 1} u_{i+1} \$^{m + 1 - |u_i| - (k_i + 1)} & \mbox{ for } i \in \{1,\ldots,n-1\}; \\
   \$^{k_1} u_1 \$^{m + 1 - |u_1| - k_1} & \mbox{ otherwise}.
  \end{array}
  \right.
 \]
 As $w_i \in \$^* u_{i+1} \$^* \cap \Sigma^{m+1}$ for $i \in \{1,\ldots,n-1\}$
 and $w_n \in \$^* u_1 \$^* \cap \Sigma^{m+1}$,
 by Lemma~\ref{lem:shuffle_2_form},
 we find
 $$
  u = h(I(w_1, \ldots, w_n) \in  \shuffle_2^n(u_2, \ldots, u_n, u_1).
 $$
 So, as $u$ was chosen arbitrarily, $\shuffle_2^n(u_1, u_2, \ldots, u_n) \subseteq \shuffle_2^n(u_2, \ldots, u_n, u_1)$.
 As the words $u_1, \ldots, u_n$ were chosen arbitrarily, we can cyclically permute them further, until
 we again reach $\shuffle_2^n(u_1, u_2, \ldots, u_n)$, which shows the equality.

 %

\item Suppose $u_1, \ldots, u_n \in \Sigma^*$.
 By Definition~\ref{lem:hom_char},
 that we have $\shuffle_1(u_1, \ldots, u_n) \in \shuffle_2^n(u_1, \ldots, u_n)$
 is obvious, as we can pad the symbols $u_i$, $i\in \{1,\ldots,n\}$,
 in the definition with additional dollar signs.
 By extension to languages, we get the first inclusion.
 
 Next, note that $\shuffle_1^n(u_1, \ldots, u_n) \in u_1 \shuffle \ldots \shuffle u_n$.
 Then, with Lemma~\ref{lem:shuffle_hom_inclusion},
 \begin{align*} 
  \{ h(\shuffle_1^n(\$^{k_1} u_1, \ldots, \$^{k_n} u_n)) \} 
  & \subseteq h(\$^{k_1} u_1 \shuffle \ldots \shuffle \$^{k_n} u_n) \\
  & \subseteq h( (\$^{k_1} u_1) \shuffle \ldots \shuffle (\$^{k_n} u_n)) \\ 
  & = u_1 \shuffle \ldots \shuffle u_n.
 \end{align*}
 So, as every word
 in $\shuffle_2^n(u_1, \ldots, u_n)$
 could be written in the form 
 \[ 
 h(\shuffle_1^n(\$^{k_1} u_1, \ldots, \$^{k_n} u_n))
 \]
 for number $k_i \ge 0$ with $i \in \{1,\ldots,n\}$,
 we find $\shuffle_2^n(u_1, \ldots, u_n) \subseteq u_1 \shuffle \ldots \shuffle u_n$.
 Hence, these equations are also true by their extensions
 to languages and we get the other inclusion.

\item By Property~\ref{prop:properties_shuffles:concat_in_shuffle_2} of Proposition~\ref{prop:properties_shuffles}
 shown previously.

\item With the previous Property~\ref{prop:properties_shuffles:Kleene_in_it_shuffle_2}
 of Proposition~\ref{prop:properties_shuffles},
 we find $\Sigma^* \subseteq \itshuffle{\Sigma}{2}$.
 Further, we have
\begin{align*} 
 \shuffle_1^n(\Sigma, \ldots, \Sigma)
  & = \{ \shuffle_1^n(x_1, \ldots, x_n) \mid x_1, \ldots, x_n \in \Sigma \} \\
  & = \{ x_1 \cdots x_n \mid x_1, \ldots, x_n \in \Sigma \} \\
  & = \Sigma^n.
\end{align*}
So $\Sigma^* \subseteq \itshuffle{\Sigma}{1}$.
As $\itshuffle{\Sigma}{i} \subseteq \Sigma^*$
for each $i \in \{1,2\}$, we have equality.

\item This is implied by Property~\ref{prop:properties_shuffles:shuffle_1_in_shuffle_2}
 of this Proposition~\ref{prop:properties_shuffles}.
 
\item By Property~\ref{prop:properties_shuffles:shuffle_1_in_shuffle_2}, both shuffle variants
 are included in the general shuffle, and the conclusion holds
 for the general shuffle. More specifically, for the (general) shuffle, choosing
 the non-erasing homomorphism $h : \Sigma^* \to a^*$ which maps
 every symbol in $\Sigma$ to $a$, 
 and using Lemma~\ref{lem:shuffle_hom_inclusion},
 for $u,v \in \Sigma^*$
 we have, as in the unary case concatenation and shuffle coincide,
 $$
  h(u \shuffle v) \subseteq h(u) \shuffle h(v) = h(u) \cdot h(v) = \{ a^{|u| + |v|} \}, 
 $$
 which, by the way, implies equality above.
 So, if $w \in u\shuffle v$, then $a^{|w|} = h(w) = a^{|u| + |v|}$,
 which implies $|w| = |u| + |v|$.
\end{enumerate}
 So, we have proven all properties.~\qed
\end{proof}

\begin{remark} 
 By the first two properties, all permutations of concatenations of arguments
 are in $\shuffle_2^n$, but this shuffle variant itself is only invariant
 under a cyclic permutation of its arguments.
 Note that Example~\ref{ex:shuffle_2_abbc} shows that it is not invariant
 under arbitrary permutations of its arguments.
 For $n = 2$, where it equals the literal shuffle by Lemma~\ref{lem:n_2_equals_init_and_literal_shuffle},
 as interchanging is the only non-trivial permutation, which is a cyclic one,
 this product is commutative, as was noted in~\cite{DBLP:journals/actaC/Berard87,DBLP:journals/tcs/Berard87}.
 But as shown above, this property only extends to cyclic permutation for more than two arguments.
\end{remark}

With both iterated variants we can describe languages that are not context-free,
as shown by the next proposition. Comparing Proposition~\ref{prop:iterated_variant_non_cf}
with Proposition~\ref{prop:finite_binary_iterative}, 
we see that these operations are more powerful than the iterated initial literal shuffle
from~\cite{DBLP:journals/actaC/Berard87,DBLP:journals/tcs/Berard87}, in
the sense that we can leave the family of regular languages for finite input languages.

\begin{propositionrep} 
\label{prop:iterated_variant_non_cf}
 $
  \itshuffle{(abc)}{1}
   = \itshuffle{(abc)}{2} \cap a^* b^* c^* 
   = \{ a^m b^m c^m \mid m \ge 0 \}. 
 $
\end{propositionrep}
\begin{proof} 
 Note that $\shuffle_1^n(abc,\ldots,abc) = a^n b^n c^n$,
 so $\{ a^m b^m c^m \mid m \ge 0 \} = \itshuffle{(abc)}{1}$.
 By Proposition~\ref{prop:properties_shuffles},
 $\{ a^m b^m c^m \mid m \ge 0 \} \subseteq \itshuffle{(abc)}{2} \cap a^* b^* c^*$.
 By applying the projection homomorphism $\pi_a : \Sigma^* \to \{a\}^*$
 for $a \in \Sigma$ given by $\pi_a(a) = a$ and $\pi_a(x) = \varepsilon$
 for $x \in \Sigma \setminus \{a\}$
 and Lemma~\ref{lem:shuffle_hom_inclusion},
 we find that, for $u,v \in \Sigma^*$,
 for any $w \in u \shuffle v$ we have $|w|_a = |u|_a + |u|_a$ for any $a \in \Sigma$.
 Hence, with Proposition~\ref{prop:properties_shuffles},
 as $\shuffle_2^n(u_1, \ldots, u_n) \subseteq u_1 \shuffle \ldots \shuffle u_n$ for $u_1, \ldots, u_n \in \Sigma^*$,
 we can deduce that $|\shuffle_2^n(abc, \ldots, abc)|_x = n$
 for $x \in \{a,b,c\}$.
 So, $\shuffle_2^n(abc, \ldots, abc) \cap a^* b^* c^* = \{ a^n b^n c^n \}$
 and we find that $\itshuffle{(abc)}{2} \cap a^* b^* c^* = \{ a^m b^m c^m \mid m \ge 0 \}$.~\qed
\end{proof}

\section{Closure Properties}
\label{sec:closure_properties}

We first show that, when adding the full trio operations, the iterated version
of our shuffle variants are as powerful as the shuffle, or as powerful
as the iterated variants of the binary versions of the initial literal and literal shuffle introduced in~\cite{DBLP:journals/actaC/Berard87,DBLP:journals/tcs/Berard87}
by Proposition~\ref{prop:trio_shuffle_lit_init}.
But before, and for establishing our closure properties,
we state the next result.

\begin{lemmarep} 
\label{lem:intreg} 
 Let $\mathcal L$ be a family of languages.
 If $\mathcal L$ is closed under intersection with regular
 languages, isomorphic mappings and (general) shuffle, then $\mathcal L$
 is closed under each shuffle variant $\shuffle_i^n$
 for $i \in \{1,2\}$.
\end{lemmarep}
\begin{proof}
 Suppose $\mathcal L$ is closed under intersection with regular
 languages, isomorphic mappings and shuffle.
 For the sake of notational simplicity, we shall give the proof, which
 is perfectly general, in the special case $n = 3$.
 Let $L_1, L_2, L_3 \in \mathcal L$ be three languages over the disjoint
 alphabets $\Sigma_1, \Sigma_2, \Sigma_3$.
 This assumption is no restriction, as by using isomorphisms of alphabets, i.e.,
 homomorphism induced by bijections between two alphabets, which are themselves
 isomorphisms,
 we can always map languages to corresponding languages over disjoint alphabets, perform
 our operations, which are preserved under isomorphisms, and finally map them back.
 First, we have
 \[
   \shuffle_1^3(L_1, L_2, L_3) = (L_1 \shuffle L_2 \shuffle L_3) \cap U,
 \]
 where 
 $
  U = (\Sigma_1 \Sigma_2 \Sigma_3)^* ( (\Sigma_{1} \Sigma_{2})^* \cup (\Sigma_{1} \Sigma_{3})^* \cup (\Sigma_{2} \Sigma_{3})^* )
      (\Sigma_1^* \cup \Sigma_2^* \cup \Sigma_3^*).
 $
 For the other shuffle operation, set (for all cases that words in $L_1, L_2, L_3$ can ``overlap'')
 \begin{align*} 
  S  = & \, \Sigma_1^* (( \Sigma_1 \Sigma_2)^* \cup (\Sigma_1 \Sigma_3)^*)(\Sigma_1 \Sigma_2 \Sigma_3)^*(( \Sigma_1 \Sigma_2)^* \cup (\Sigma_1 \Sigma_3)^*)\Sigma_1^* \\
     \cup  & \, \Sigma_2^* (( \Sigma_1 \Sigma_2)^* \cup (\Sigma_2 \Sigma_3)^*)(\Sigma_1 \Sigma_2 \Sigma_3)^*(( \Sigma_1 \Sigma_2)^* \cup (\Sigma_2 \Sigma_3)^*)\Sigma_2^* \\ 
     \cup & \,  \Sigma_3^* (( \Sigma_1 \Sigma_3)^* \cup (\Sigma_2 \Sigma_3)^*)(\Sigma_1 \Sigma_2 \Sigma_3)^*(( \Sigma_1 \Sigma_3)^* \cup (\Sigma_2 \Sigma_3)^*)\Sigma_3^*, \\
  T = &  \,\Sigma_1^*(\Sigma_1 \Sigma_2)^* \Sigma_1^* (\Sigma_1 \Sigma_3)^* \Sigma_1^* \cup \Sigma_1^*(\Sigma_1 \Sigma_3)^* \Sigma_1^* (\Sigma_1 \Sigma_2)^* \Sigma_1^*\\
    \cup & \, \Sigma_2^*(\Sigma_1 \Sigma_2)^* \Sigma_2^* (\Sigma_2 \Sigma_3)^* \Sigma_2^* \cup \Sigma_2^*(\Sigma_2 \Sigma_3)^* \Sigma_2^* (\Sigma_1 \Sigma_2)^* \Sigma_2^* \\
    \cup & \, \Sigma_3^*(\Sigma_1 \Sigma_3)^* \Sigma_3^* (\Sigma_2 \Sigma_3)^* \Sigma_3^* \cup \Sigma_3^*(\Sigma_2 \Sigma_3)^* \Sigma_3^* (\Sigma_1 \Sigma_3)^* \Sigma_3^*, \\
  R =    & \, \Sigma_1^*(\Sigma_1 \Sigma_2)^* (\Sigma_1 \Sigma_2 \Sigma_3)^* (\Sigma_2 \Sigma_3)^* (\Sigma_2^* \cup \Sigma_3^*) \\ 
    \cup & \, \Sigma_2^*(\Sigma_1 \Sigma_2)^* (\Sigma_1 \Sigma_2 \Sigma_3)^* (\Sigma_1 \Sigma_3)^* (\Sigma_1^* \cup \Sigma_3^*) \\
    \cup & \, \Sigma_3^*(\Sigma_1 \Sigma_3)^* (\Sigma_1 \Sigma_2 \Sigma_3)^* (\Sigma_2 \Sigma_3)^* (\Sigma_2^* \cup \Sigma_3^*) \\
    \cup & \, \Sigma_1^*(\Sigma_1 \Sigma_3)^* (\Sigma_1 \Sigma_2 \Sigma_3)^* (\Sigma_2 \Sigma_3)^* (\Sigma_2^* \cup \Sigma_3^*) \\ 
    \cup & \, \Sigma_2^*(\Sigma_2 \Sigma_3)^* (\Sigma_1 \Sigma_2 \Sigma_3)^* (\Sigma_1 \Sigma_3)^* (\Sigma_1^* \cup \Sigma_3^*) \\
    \cup & \, \Sigma_3^*(\Sigma_2 \Sigma_3)^* (\Sigma_1 \Sigma_2 \Sigma_3)^* (\Sigma_1 \Sigma_2)^* (\Sigma_1^* \cup \Sigma_2^*) \\
  P = & \, \bigcup_{\substack{\pi : \{1,2,3\}\to\{1,2,3\} \\ \mbox{ a permutation }}} \Sigma_{\pi(1)}^* \Sigma_{\pi(2)}^* \Sigma_{\pi(3)}^*.
 \end{align*}
 Then, we have
 \[
 \shuffle_2^3(L_1, L_2, L_3) = (L_1 \shuffle L_2 \shuffle L_3) \cap (S \cup T \cup R \cup P).
 \] 
 Hence, the can write our shuffle variants with the general shuffle and intersection
 by a suitable regular language.~\qed
\end{proof}


With the previous lemma, we can derive the next result.

\begin{propositionrep} 
\label{prop:trio_shuffles_closures}
 Let $\mathcal L$ be a full trio. The following properties are equivalent:
 \begin{enumerate}
 \item $\mathcal L$ is closed under shuffle.
 \item $\mathcal L$ is closed under $\shuffle_i^n$
 for some $i \in\{1,2\}$ and $n \ge 2$.
 \end{enumerate}
\end{propositionrep}
\begin{proof}
 First suppose $\mathcal L$ is a full trio that is closed under shuffle.
 Then, by Lemma~\ref{lem:intreg} it is closed
 under each operation $\shuffle_n^i$ for $i \in \{1,\ldots, n\}$.
 Conversely, suppose $\mathcal L$ is closed under $\shuffle_i^n$
 for $i \in \{1,\ldots, 4\}$ with $n \ge 2$.
 We distinguish several cases. Note~\cite{DBLP:journals/actaC/Berard87,DBLP:journals/tcs/Berard87} 
     that, for the homomorphism $h : (\Sigma \cup \{\$\})^* \to \Sigma^*$
     given by $h(x) = x$ for $x \in \Sigma$ and $h(\$) = \varepsilon$ with $\$ \notin \Sigma$,
     we have
     \begin{equation}\label{eqn:init_shuffle_hom_shuffle}
      L_1 \shuffle L_2 = h(\shuffle_1^2(h^{-1}(L_1), h^{-1}(L_2))). 
     \end{equation}
 Let $i \in \{1,2\}$ and $L_1, L_2 \in \mathcal L$.
    We have $\shuffle_i^n(L_1, L_2, \{ \varepsilon \}, \ldots, \{\varepsilon\}) = \shuffle_i^2(L_1, L_2)$
    and the results follow by Lemma~\ref{lem:n_2_equals_init_and_literal_shuffle}
    and Proposition~\ref{prop:trio_shuffle_lit_init}.
   
 So, in each case, we can represent the shuffle, which gives our claim.~\qed
\end{proof}

In a full trio, we can express the iterated shuffle
with our iterated versions of the $n$-ary shuffle variants. 

\begin{propositionrep}
\label{prop:hom_shuffles}
 Let $L \subseteq \Sigma^*$ be a language, $\$ \notin \Sigma$,
 and $h : (\Sigma \cup \$)^* \to \Sigma^*$ the homomorphism
 defined by: $h(x) = x$ if $x \in \Sigma$, $h(\$) = \varepsilon$.
 Then, for $i \in \{1,2\}$,
 \[
  L^{\shuffle,*} = h(\itshuffle{h^{-1}(L)}{i}).
 \]
\end{propositionrep}
\begin{proof}

 By Proposition~\ref{prop:properties_shuffles},
 for any $i \in \{1,2\}$,
 $\itshuffle{h^{-1}(L)}{i} \subseteq (h^{-1}(L))^{\shuffle,*}$.
 Hence, using Lemma~\ref{lem:shuffle_hom_inclusion},
 we can conclude $h(\itshuffle{h^{-1}(L)}{i}) \subseteq h((h^{-1}(L))^{\shuffle,*})
 \subseteq (h(h^{-1}(L)))^{\shuffle,*} \subseteq L^{\shuffle,*}$.
 Similarly, by Proposition~\ref{prop:properties_shuffles},
 we find that for $i = 2$
 we have
 $$
  \itshuffle{h^{-1}(L)}{1} \subseteq  \itshuffle{h^{-1}(L)}{i} \subseteq h^{-1}(L)^{\shuffle,*}
 $$
 and so, with Lemma~\ref{lem:shuffle_hom_inclusion},
 $$ 
  h(\itshuffle{h^{-1}(L)}{1}) \subseteq  h(\itshuffle{h^{-1}(L)}{i}) \subseteq L^{\shuffle,*}.
 $$
 Hence, if we can show $L^{\shuffle,*} \subseteq h(\itshuffle{h^{-1}(L)}{1})$,
 the claim follows.
 So, the rest of the proof is dedicated to showing this inclusion.
 Let $u \in u_1 \shuffle \ldots \shuffle u_n$
 with $n \ge 1$ and $u_i \in L$ for $i \in \{1,\ldots, n\}$.
 Then,
 $$
  u = \prod_{j=1}^m (u_{1,j} u_{2,j} \cdots u_{n,j})
 $$
 for some $m \ge 0$ such that 
 $u_i = u_{i,1} \cdots u_{i,m}$ with $u_{i,j} \in \Sigma \cup \{\varepsilon\}$
 for $j \in \{1,\ldots, m\}$ and $i \in \{1,\ldots, n\}$. Set
 $$
  v_{i,j} = \left\{
  \begin{array}{ll}
   u_{i,j} & \mbox{if } u_{i,j} \in \Sigma; \\ 
   \$      & \mbox{if } u_{i,j} = \varepsilon;
  \end{array}\right.
 $$
 and $v = \prod_{i=1}^m (v_{1,i} v_{2,i} \cdots v_{n,i})$.
 Then $v \in h^{-1}(u)$ and $h(v) = u$. But also
 \[ 
 v = \shuffle_1^n(v_{1,1} \cdots v_{1,m}, \ldots, v_{n,1} \cdots v_{n,m}).
 \]
 Observe that here
 $$\shuffle_1^n(v_{1,1} \cdots v_{1,m}, \ldots, v_{n,1} \cdots v_{n,m})
 = I(v_{1,1} \cdots v_{1,m}, \ldots, v_{n,1} \cdots v_{n,m})$$
 as these words all have the same length. 
 Furthermore, $$h(v_{i,1} \cdots v_{i,m}) = u_{i,1} \cdots u_{i,m} = u_i$$
 and $v_{i,1} \cdots v_{i,m} \in h^{-1}(u_i)$.
 So, 
 $$ 
 u \in h(\shuffle_1^n(h^{-1}(u_1), \ldots, h^{-1}(u_n))) \subseteq h(\itshuffle{h^{-1}(L)}{1}).
 $$
 Hence, we have the inclusion $L^{\shuffle,*} \subseteq h(\itshuffle{h^{-1}(L)}{1})$
 and we are done
 as, by Proposition~\ref{prop:properties_shuffles},
 $\itshuffle{L}{1}\subseteq \itshuffle{L}{2}$.~\qed
\end{proof}

It is well-known that the regular languages are closed
under shuffle~\cite{BrzozowskiJLRS16}.
Also, the context-sensitive languages are closed
under shuffle~\cite{DBLP:journals/ipl/Jedrzejowicz83}.
A full trio is closed under intersection if and only if it
is closed under shuffle~\cite{Ginsburg75}.
As the recursively enumerable languages are a full trio that is closed under
intersection, this family of languages is also closed under shuffle.
As the context-free languages are a full trio that is not closed
under intersection~\cite{HopUll79}, it is also not closed
under shuffle. The last result is also implied by Proposition~\ref{prop:iterated_variant_non_cf}
and Proposition~\ref{prop:trio_shuffles_closures}.
The recursive languages are only closed under non-erasing homomorphisms,
so we could not reason similarly. Nevertheless, this family 
of languages is closed under shuffle. 

\begin{toappendix}
 A Turing machine $T$ is specified by
 a $6$-tuple $T = (Q, \Gamma, \square, \Sigma, q, F)$,
 where
 \begin{enumerate}
 \item $Q$ is the set of states,
 \item $\Gamma$ is the tape alphabet,
 \item $\square \in \Gamma$ denotes the blank symbol,
 \item $\Sigma \subseteq \Gamma \setminus \{ \square\}$ is the input alphabet,
 \item $q \in Q$ the start state, and
 \item $F \subseteq Q$ the set of final states.
 \end{enumerate}
 We will also freely switch between other models, for example using
 more than one working tape. See~\cite{HopUll79}
 for more details.
\end{toappendix}

\begin{propositionrep}
\label{prop:rec_closed_shuffle}
 The family of recursive languages is closed under shuffle.
\end{propositionrep}
\begin{proof} 

 Let $T_1 = (Q_1, \Gamma_1, \square_1, \Sigma, \delta, q_1, F_1)$
 and $T_2 = (Q_2, \Gamma_2, \square_2, \Sigma, \delta, q_2, F_2)$
 be two Turing machines deciding languages
 $L_1$ and $L_2$.
 We construct a two Turing machine $T$
 with two working tapes and on read-only input tape.
 Using additional working tapes does not increase the computational
 power~\cite{HopUll79}, and stipulating a read-only input tape
 is no restriction, as we can copy everything to a working
 tape and then operate like a one tape machine.
 Our machine $T$ non-deterministically reads the input tape,
 writing the current symbol either on the first, or the second tape.
 After it has read the input entirely, it first simulates $T_1$
 one the first tape, then it simulated $T_2$
 on the second tape.
 It accepts if and only if both accept, and rejects otherwise.
 Note that this construction would also work for recursively
 enumerable languages, then $T_1$ and $T_2$
 are only recognizers. We get a non-deterministic recognizing
 machine for the shuffle, for which a halting and accepting computation
 exists if and only the input word is the shuffle of two words
 form $L_1$ and $L_2$.
 Another argument, more specific to deciding machines,
 would be to try all words of length smaller than the input words, compute
 their shuffle and test if it contains the input word.~\qed
\end{proof}

We now state the closure properties of the families of regular,
context-sensitive, recursive and recursively enumerable languages. 

\begin{propositionrep}
\label{prop:chomsky_hierarchy_closures}
 The families of regular, context-sensitive, recursive
 and recursively enumerable languages are closed
 under $\shuffle_i^n$ for $i \in \{1,2\}$.
 Furthermore, the families of context-sensitive,
 recursive and recursively enumerable languages are closed under
 the iterated versions, i.e., if $L$ is context-sensitive, recursive
 or recursively enumerable, then $L^{\shuffle_i,\circledast}$, $i \in \{1,2\}$,
 is context-sensitive, recursive or recursively enumerable, respectively.


    

\end{propositionrep}
\begin{proofsketch}
 The closure of all mentioned language families under
 $\shuffle_i^n$ with $i \in \{1,2\}$ is implied by Lemma~\ref{lem:intreg},
 as they are all closed under intersection with regular languages
 and shuffle by Proposition~\ref{prop:rec_closed_shuffle}
 and the considerations before this statement.
 Now, we give a sketch for the iterated variant $\itshuffle{L}{1}$.
 
 Let $M = (\Sigma, \Gamma, Q, \delta, s_0, F)$ be a Turing machine for $L$.
 The following construction will work for all language classes. More specifically, 
 if given a context-sensitive, recursive or recursively enumerable
 language $L$ with an appropriate machine $M$, it could be modified to give a
 machine that describes a language in the corresponding class, but the basic idea
 is the same in all three cases.
 Recall that the context-sensitive languages could be
 characterized by linear bounded automata~\cite{HopUll79}.
 
 We construct a $3$-tape Turing machine, with one input tape, that simulates $M$ and has three working tapes.
 Intuitively,
 \begin{enumerate}
 \item the input tape stores the input $u$;
 \item on the first working tape, the input is written in a decomposed way, and on certain
  parts, the machine $M$ is simulated;
 \item on the second working tape, for each simulation run of $M$, a state of $M$ is saved;
 \item the last working tape is used to guess and store a number $0 < n \le |u|$.
 \end{enumerate}
 
 We sketch the working of the machine. First, it non-deterministically
 guesses a number $0 < n \le |u|$ and stores it on the last tape.
 Then, it parses the input $u$ in several passes, each pass
 takes $0 < k \le n$ symbols from the front of $u$
 and puts them in an ordered way on the second working tape, and, non-deterministically,
 decreases $k$ or does not decrease $k$.
 More specifically, on the second working tape, the machine writes
 a word, with a special separation sign $\#$,
 \[
  \# u_1 \# u_2 \# \cdots \# u_n \#
 \]
 where $\shuffle_1^n(u_1, \ldots, u_n)$ equals the input parsed so far.
 When the input word is completey parsed, it simulates $M$
 to check if each word $u_i$ on the second working tape
 is contained in $L$.\qed
\end{proofsketch}
\begin{proof}
 The closure of all mentioned language families under
 $\shuffle_i^n$ with $i \in \{1,2\}$ is implied by Lemma~\ref{lem:intreg},
 as they are all closed under intersection with regular languages
 and shuffle by Proposition~\ref{prop:rec_closed_shuffle}
 and the considerations before this statement.
 Now, for the iterated variants $\itshuffle{L}{i}$, $i \in \{1,\ldots, 4\}$.
 
 \medskip 
 
 \noindent\underline{The case $i = 1$}.
 
 \medskip 
 
 We give a general construction, and then argue that it works 
 also for context-sensitive and recursive languages.
 Let $M = (\Sigma, \Gamma, Q, \delta, s_0, F)$ be a Turing machine for $L$.
 We construct a $3$-tape Turing machine, with one input tape, that simulates $M$ and has three working tapes.
 Intuitively,
 \begin{enumerate}
 \item the input tape stores the input $u$;
 \item on the first working tape, the input is written in a decomposed way, and on certain
  parts, the machine $M$ is simulated;
 \item on the second working tape, for each simulation run of $M$, a state of $M$ is saved;
 \item the last working tape is used to guess and store a number $0 < n \le |u|$.
 \end{enumerate}
 
 
 %
 % ...
 %

 %
 %
 %

 
 Now, to be more specific:
 Suppose $u \in \Sigma^*$ is the input word. If $|u| = 0$, then accept; so 
 assume $|u| > 0$ from now on.
 \begin{enumerate}
 \item Non-deterministically guess a number of operands $0 < n \le |u|$ 
  and write the result on the last working tape.

 \item We read $u$ from left to right, always moving the reading head one step to the right.
  First, write the first $n$ symbols of $u$ onto the first working tape, and put the separator sign $\#$
  between them and at the start and end.
  We will construct $n$ sequences between these separators, and after this first
  step we have constructed $n$ sequences with one symbol each. 
  To be more specific, if $u = x_1 \cdots x_m$ with $m = |u|$ and $x_i \in \Sigma$ for $i \in \{1,\ldots, m\}$,
  then, after this first pass, the content on the first working tape is
  $$
   \# x_1 \# x_2 \# \cdots \# x_n \#.
  $$
  Also, at each step, $0 \le k \le n$ of these sequences between the separator signs will be active, meaning we can put other symbols at their end.
  We will mark the inactive ones, either by putting a special symbol at the end, or changing the separator after these, or mark the last (top) symbol.
  Additionally, we will use a counter $1 \le j \le n$, which could also be stored on the last working tape.
  Also, in the process, we will scan the rest of the input, i.e., $x_{n+1} \cdots x_m$,
  sequentielly from left to right.

  First, set $j = 1$ and, as none is marked, all sequences are active. 
  Until we have not reached the end of the input, repeat the next steps:
  \begin{enumerate}
  \item[(i)] If the $j$-th sequence is active, either, non-deterministically,
    put the next symbol of the input at the end of it and move the reading head one step to the right,
    or, but only if $k > 1$, i.e., we have at least two active sequences left, 
     mark the $j$-th sequence as inactive.
  \item[(ii)] If the $j$-th sequence is inactive, do nothing, i.e., those sequences are skipped.
  \item[(iii)] Finally, if (i) or (ii) was executed and if $j = n$, then set $j = 1$ and start the process again; otherwise increase $j$ by one.
  \end{enumerate}
  In between these steps, maintain the separation between the sequences by copying everything one step to the right, if a symbol is put at the end of a specific sequence.
  
 \item After these steps, the input is decomposed into multiple sequences on the first working tape, i.e.
  it has the form
  $$
   \# v_1 \# v_2 \# \cdots \# v_n \#
  $$
  where the $v_i$ are scattered subsequences\footnote{A word $u = x_1 \cdots x_n$ with $x_i \in \Sigma$, $i \in \{1,\ldots, n\}$, is a scattered subsequence
  of $v$ if $v \in \Sigma^* x_1 \Sigma^* x_2 \Sigma^* \cdots \Sigma^* x_n \Sigma^*$.} 
  of the input that are produced in the previous step, and $v_1 v_2 \cdots v_n$
  is a specific permutation of the input $u$.
  Then, simulate $M$ on each sequence $v_i$ for $i \in \{1,\ldots, n\}$, 
  i.e., using this sequence as input.
  But each time a separator sign
  is read, move everything to the right if the head moves to the right, and insert a new separator sign; 
  and treat the separator sign like a blank symbol for $M$ (this step is not necessary if $M$ moves its reading 
  head only inside the input portion and at most one step to the left and right). 
  If we have simulated $M$ for each sequences, and each time $M$ accepts,
  then $u \in \shuffle^1_n(v_1, \ldots, v_n)$
  and our constructed machine accepts. 
  Also, if $u \in \shuffle^1_n(v_1, \ldots, v_n)$, then at least one accepting
  run exists.
 \end{enumerate}
   
 Now, note the following.
 
 \begin{enumerate}
 \item If the input language is recursively enumerable, we have constructed a recognizing
 non-deterministic Turing machine, hence the shuffled language
 is also recursively enumerable.
     
 \item If the input language is recursive, in the last simulation step, the machines always
  terminate. Hence, the constructed machine also terminates, and it is actually
  a deciding machine for the shuffled language.
  
 \item Suppose the input machines are linear bounded automata as defined in~\cite{HopUll79}, i.e.,
  using only space $|v|$ for each input $v$. This implies that 
 in the simulation step they do not read symbols outside of $v_1, \ldots, v_n$.
  From the multiple tapes, we can construct a single tape machine (also incorporating the reading tape)
  with only a linear blow up in space~\cite{HopUll79}. The counters, and the states, we have
  to store on the working tapes, are bounded by the input. So, the simulation, 
  as additionally we only need a copy of the states and transition rules of $M$ in the constructed machine,
  could be performed in linear space.
  Hence, in total, we only need
  linear space on each working tape, which gives us in total that we only need linear space.
  Hence, by standard results~\cite{HopUll79}, the language we get is accepted
  by a linear bounded automaton and so context-sensitive.
 \end{enumerate}

  \medskip 
 
 \noindent\underline{The case $i = 2$}.
 
 \medskip 
 
 We use Lemma~\ref{lem:shuffle_2_form} and modify the machine for 
 the case $i = 1$.
 Let $N = |u| + 1$.
 First, perform the same steps as this machine, until 
 \[ 
   \# x_1 \# x_2 \# \cdots \# x_n \#.
 \] 
 is written on the tape. Then, non-deterministically,
 guess numbers $0 \le k_1, \ldots, k_n \le N$
 and write
 \[  
     \# \$^{k_1} x_1 \# \$^{k_2} x_2 \# \cdots \# \$^{k_n} x_n \#.
 \]
  on the working tape. Then, modify $M$
  such that the dollar sign $\$$ is skipped when read as an input, without
  a state change.
  Then, run the same steps as the algorithm for $i = 1$.
  By Lemma~\ref{lem:shuffle_2_form}, the input is accepted if and only 
  if it is in $u_2^n(L, \ldots, L)$.
  
  Similar to the case $i = 1$, we see that all closure properties are
 satisfied.
 
 \medskip

These constructions show closure under all iterated shuffle variants.  
We have used closure properties to show closure under $\shuffle_i^n$
for $i \in \{1,2\}$ above. But note, we could also use the above constructions,
but with $n$ fixed instead of a guessed value.~\qed
\end{proof}

Lastly, we can characterize the family of non-empty finite languages
using $\shuffle_1^n$.

\begin{propositionrep}
 \label{prop:fin_lang}
 The family of non-empty finite languages is the smallest family
 of languages $\mathcal L$
 such that
 \begin{enumerate}
 \item $\{ w \} \in \mathcal L$ for some word $w \ne \varepsilon$ with all symbols in $w$ distinct, i.e., $w = a_1 \cdots a_m$ with $a_i \ne a_j$ for $1 \le i \ne j \le m$ and $a_i \in \Sigma$ for $i \in \{1,\ldots, m\}$,
 
 \item  closed under union,
 
 \item closed under homomorphisms $h : \Sigma^* \to \Gamma^*$
 such that $|h(x)| \le 1$ for $x \in \Sigma$,
 
 \item closed under $\shuffle_1^n$ for some $n \ge 2$.
 \end{enumerate}
\end{propositionrep}
\begin{proof}
 Let $\mathcal L$ be a family of language admitting the stated closure properties.
 First, let us note that we can assume $\mathcal L$
 is closed under $\shuffle_1^2$, 
 as $\shuffle_1^2(u,v) = \shuffle_1^2(u,v,\varepsilon,\ldots, \varepsilon)$
 and $\{\varepsilon\} \in \mathcal L$ by applying the homomorphism 
 $h : \Sigma^* \to \Gamma^*$ given by $h(x) = \varepsilon$ for $x \in \Sigma$ to $w$ with $\{w\} \in \mathcal L$.
 Clearly, $\mathcal L$ contains only finite languages.
 The closure properties of $\mathcal L$ imply that it suffices
 to prove that the singleton languages consisting of an arbitrary word $w = a_1 \cdots a_n$,
 where all symbols $a_1, \ldots, a_n$ are distinct, 
 are in $\mathcal L$.
 We prove this by induction on $n$. By choosing an appropriate homomorphism, every
 word of length at most $|w|$ could written as a homomorphic image of $w$,
 so the assertion holds for $n \in \{0,1,\ldots, |w|\}$.
 For $n > |w|$, we have
 $\{ a_1 a_2 \cdots a_n \} = \shuffle_1^2(a_1, a_2 \cdots a_n)$.\qed 
\end{proof}

And, without closure under any homomorphic mappings.

\begin{propositionrep}
 The family of non-empty finite languages is the smallest family
 of languages $\mathcal L$ such that
 (1) $\{ \{\varepsilon \} \} \cup \{ \{a\} \mid a \in \Sigma \} \subseteq \mathcal L$ and
 which is (2) closed under union and $\shuffle_1^n$ for some $n \ge 2$.
\end{propositionrep}
\begin{proof} 
 For any word $w = a_1 a_2 \cdots a_n$ with $a_i \in \Sigma$ for $i \in \{1,\ldots, n\}$
 and such that the symbols do not have to be distinct, 
 we have $\{ w \} = \shuffle_1^2(a_1, a_2 \cdots a_n)$.
 Then, the proof could by carried out with similar arguments as in the proof 
 of Proposition~\ref{prop:fin_lang}. 
\end{proof}

\section{Computational Complexity}


    
    


Here, we consider various decision problems for both shuffle variants
motivated by similar problems for the ordinary shuffle operation~\cite{DBLP:journals/jcss/WarmuthH84,stockmeyer1973word,DBLP:journals/tcs/BerglundBB13,Ogden78}.
It could be noted that all problems considered are tractable
when considered for $\shuffle_1$. However, for $\shuffle_2$, most problems
considered are tractable except one that is \NP-complete. Hence, the ability to vary the starting
positions of different words when interlacing them consecutively, in an alternating fashion,
seems to introduce computional hardness.
For $\shuffle_1$, we find the following:

\begin{propositionrep}
 Given $L \subseteq \Sigma^*$ represented by a non-deterministic\footnote{In a non-deterministic automaton
 the transitions are represented by a relation instead of a function, see~\cite{HopUll79}.
 } automaton and words $w_1, \ldots, w_n \in \Sigma^*$,
 it is decidable in polynomial time
 if $\shuffle_1^n(w_1, \ldots, w_n) \in L$.
\end{propositionrep}
\begin{proof}
 For the words $w_1, \ldots, w_n$, put all the first symbols in front, until one word has no more left,
 in which case proceed in the same manner with the rest.
 At the end, we check if the word is in $L$, which could be done in polynomial time
 if $L$ is given by an automaton (even for a non-deterministic automaton)
 and answer yes if some word in $L$ agrees with the result, and no otherwise.~\qed
\end{proof}

\todo{auch für endliches $L$ formulieren.}
\begin{propositionrep}
 Given words $w \in \Sigma^*$ and $v \in \Sigma^*$,
 it is decidable in polynomial time
 if $w \in \itshuffle{\{v\}}{1}$.
\end{propositionrep}
\begin{proof}
 By Proposition~\ref{prop:properties_shuffles},
 if $w \in \shuffle_1^n(v,\ldots, v)$ then $|w| = n|v|$.
 So, if this is not the case, we can immediately reject.
 Otherwise, we can divide $|w|$ by $|v|$
 to compute the necessary $n$ with $|w| = n|v|$.
 Then, if we write $v = x_1 \cdots x_m$
 with $m = |v|$ and $x_i \in \Sigma$ for $i \in \{1,\ldots,m\}$,
 then $\shuffle_1^n(v,\ldots, v) = x_1^n x_2^n \cdots x_m^n$.
 Hence, we only have to check
 if $w$ equals this word, which could be performed in polynomial time.~\qed
\end{proof}

The non-uniform membership for a languages $L \subseteq \Sigma^*$ 
is the computational problem to decide for a given word $w \in \Sigma^*$
if $w \in L$.
In~\cite{Ogden78} it was shown that the shuffle of two deterministic context-free languages
can yield a language that has an \NP-complete non-uniform membership problem. 
This result was improved in~\cite{DBLP:journals/tcs/BerglundBB13} by showing that there even
exist linear deterministic context-free languages whose shuffle gives an intractable
non-uniform membership problem. 

Next, we show that for the initial literal and the literal shuffle, 
this could not happen if the original languages have a tractable membership problem,
which is the case for context-free languages~\cite{HopUll79}.

\begin{proposition}
 Let $U, V \subseteq \Sigma^*$ be languages whose membership problem
 is solvable in polynomial time. Then, also the membership problems
 for $U \shuffle_1 V$ and $U \shuffle_2 V$
 are solvable in polynomial time.
\end{proposition}
\begin{proof}
 Let $w$ be a given word and write $w = w_1 \cdots w_n$
 with $w_i \in \Sigma$ for $i \in \{1,\ldots,n\}$.
 
 Then, to check if $w \in U\shuffle_2 V$,
 we try all decompositions $w = xyz$
 with $x,y,z \in \Sigma^*$ and $|y|$ even.
 For $w = xyz$, write $y = y_1 \cdots y_{2n}$
 with $y_i \in \Sigma$ and $n \ge 0$.
 Then test if $xy_1y_{3} \cdots y_{2n-1} \in U$
 and $y_2 \cdots y_{2n}z \in V$, or $y_1y_{3} \cdots y_{2n-1}z \in U$
 and $xy_2 \cdots y_{2n} \in V$,
 or $xy_1 y_3 \cdots y_{2n-1} z \in U$
 and $y_2 \cdots y_{2n} \in V$,
 or $y_1 y_3 \cdots y_{2n-1} \in U$ and $xy_2 \cdots y_{2n}z \in V$
 As $U \shuffle_2 V = V \shuffle_2 U$
 this is sufficient to find out if $w \in U\shuffle_2 V$.

 For $U\shuffle_1 V$, 
 first check if $w \in U$ and $\varepsilon \in V$,
 or if $\varepsilon \in U$ and $w \in V$.
 If neither of the previous checks give a YES-answer,
 then try all decompositions $w = yz$
 with $y = y_1 \cdots y_{2n}$ for $y_i \in \Sigma$ and $n > 0$.
 Then, test if $y_1 y_3 \cdots y_{2n-1}z \in U$
 and $y_2 \cdots y_{2n} \in V$,
 or if $y_1 y_3 \cdots y_{2n-1}z \in V$
 and $y_2 \cdots y_{2n} \in U$. If at least one of these tests gives
 a YES-answer, we have $w \in U \shuffle_1 V$, otherwise $w \notin U \shuffle_1 V$.
 
 In all cases, only polynomially many tests were necessary.\qed
\end{proof}

A similar procedure could be given for any fixed number $n$\todo{hinschreiben?}
and $L_1, \ldots, L_n \subseteq \Sigma^*$
to decide the membership problem for $\shuffle_i^n(L_1, \ldots, L_n)$, $i \in \{1,2\}$
in polynomial time.

Lastly, the following is an intractable problem for the second shuffle variant.

\begin{propositionrep}
\label{prop:np_problem}
  Suppose $|\Sigma| \ge 3$. Given a finite language $L \subseteq \Sigma^*$ represented by a deterministic automaton
  and words $w_1, \ldots, w_n \in \Sigma^*$,
 it is $\NP$-complete to decide if $\shuffle_2^n(w_1, \ldots, w_n) \cap L \ne \emptyset$.
\end{propositionrep}
\begin{proofsketch}
 We give the basic idea for the hardness proof. Similarly as in~\cite{DBLP:journals/jcss/WarmuthH84}
 for the corresponding problem in case of the ordinary shuffle and a single word $L = \{w\}$ as input,
 we can use a reduction from {\sc 3-Partition}. This problem is known
to be strongly \NP-complete, i.e., it is \NP-complete
even when the input numbers are encoded in unary~\cite{garey1979computers}.
 
 \begin{quote}
    {\sc 3-Partition} \\
    \emph{Input:} A sequence of natural numbers $S = \{ n_1, \ldots, n_{3m} \}$
     such that $B = (\sum_{i=1}^{3m} n_i) / m \in \mathbb N_0$
     and for each $i$, $1 \le i \le 3m$, $B/4 < n_i < B/2$.
 
     \emph{Question:} Can $S$ be partitioned into $m$ disjoint subsequences $S_1, \ldots, S_m$
     such that for each $k$, $1 \le k \le m$, $S_k$ has exactly three elements and $\sum_{n \in S_k} n = B$.
\end{quote}
 Let $S = \{n_1, \ldots, n_{3m} \}$ be an instance of {\sc 3-Partition}.
Set 
\[ 
L = \{ aaauc \in \{a,b,c\}^* \mid |u|_b = B, |u|_a = 0, |u|_c = 2 \}^m.  
\]
We can construct a deterministic automaton for $L$ in polynomial time.
Then, the given instance of {\sc 3-Partition} has a solution
if and only if \[ 
L \cap \shuffle_2^n(ab^{n_1}c, ab^{n_2}c, \ldots, ab^{n_{3m}}c) \ne \emptyset.\qed
\] 
\end{proofsketch}
\begin{proof} 

 First, we show containment in \NP, then
 we will show \NP-hardness.

\medskip 


%
%
%

\noindent\underline{Containment in \NP:} Guess a word $w$ of length at most $\max\{ |u| \mid u \in L \}$
and check if $w \in L$. Then, guess a number $N \ge |w_1| + \ldots + |w_n|$
and, for each word $w_i$, $i \in \{1,\ldots,n\}$, 
guess non-deterministically a number $0 \le j_i \le N - |w_i|$.
Set $u_i = \$^{j_i} w_i \$^{N - |w_i| - j_i}$.
Then, with the notation from Definition~\ref{lem:hom_char},
compute $h(I(u_1, \ldots, u_n))$, which is possible in polynomial time,
and compare the result with $w$. 
By Definition~\ref{lem:hom_char}, there exist
numbers such that $h(I(u_1, \ldots, u_n))$
equals $w$ if and only if $w \in \shuffle_2^n(w_1, \ldots, w_n)$.

\medskip

\noindent\underline{\NP-hardness:} Our reduction is similar to the one from~\cite{DBLP:journals/jcss/WarmuthH84}
used for the ordinary shuffle operation and a single word $L = \{w\}$.
We give a reduction from {\sc 3-Partition}, a problem known
to be strongly \NP-complete, i.e., it is \NP-complete
even when the input numbers are encoded in unary~\cite{garey1979computers}.

\begin{quote}
    {\sc 3-Partition} \\
    \emph{Input:} A sequence of natural numbers $S = \{ n_1, \ldots, n_{3m} \}$
     such that $B = (\sum_{i=1}^{3m} n_i) / m \in \mathbb N_0$
     and for each $i$, $1 \le i \le 3m$, $B/4 < n_i < B/2$.
 
     \emph{Question:} Can $S$ be partitioned into $m$ disjoint subsequences $S_1, \ldots, S_m$
     such that for each $k$, $1 \le k \le m$, $S_k$ has exactly three elements and $\sum_{n \in S_k} n = B$.
\end{quote}

 Let $S = \{n_1, \ldots, n_{3m} \}$ be an instance of {\sc 3-Partition}.
 Set
\[ 
L = \{ aaauc \in \{a,b,c\}^* \mid |u|_b = B, |u|_a = 0, |u|_c = 2 \}^m. 
\]
We can construct a deterministic automaton for $L$ in polynomial time.
Then, the given instance of {\sc 3-Partition} has a solution
if and only if \[ 
L \cap \shuffle_2^n(ab^{n_1}c, ab^{n_2}c, \ldots, ab^{n_{3m}}c) \ne \emptyset.
\]

Intuitively, the reader can image the words sitting on different aligned ``sliders''. 
Then, we move the words on these ``sliders'' such that words corresponding to a set $S_i$
of the partition are on top of each other.

If we have a partition $S_1, \ldots, S_m$, then for
each $S_k = \{ n_{i_1}, n_{i_2}, n_{i_3} \}$ with $i_1 < i_2 < i_3$, $1 \le k \le m$, let $w_k = \shuffle_1^3(ab^{n_{i_1}}c, ab^{n_{i_2}}c, ab^{n_{i_3}}c)$.
Note that the relative order of the arguments here and in $\shuffle_2^n(ab^{n_1}c, ab^{n_2}c, \ldots, ab^{n_{3m}}c)$ above\todo{genauer erläutern.}
coincide.
Then, $w_1 \cdots w_m \in L$\todo{diese gleichung zwischen beide shuffle varianten im haupttext erwähnen?}
and, as for any $r,s \ge 0$ and $u_1, \ldots, u_r, v_1, \ldots, v_s \in \Sigma^*$ we have $\shuffle_1^r(u_1, \ldots, u_r) \cdot \shuffle_1^s(v_1, \ldots, v_s) 
\in \shuffle_2^{r+s}(x_1,\ldots, x_{r+s})$
for words $\{x_1, \ldots, x_{r+s}\}$ such that $\{x_1, \ldots, x_{r+s}\}$ equals $\{ u_1, \ldots, u_r, v_1, \ldots, v_s \}$
and the relative orders are preserved, \todo{todo, wie schuffeln.}
we also have 
\[ 
 w_1 \cdots w_m \in \shuffle_2^n(ab^{n_1}c, ab^{n_2}c, \ldots, ab^{n_{3m}}c).
\]
Hence, $w_1 \cdots w_m \in L \cap \shuffle_2^n(ab^{n_1}c, ab^{n_2}c, \ldots, ab^{n_{3m}}c)$.
\todo{noch problem mit der anordnung der argumente!}

Conversely, suppose $w \in L \cap \shuffle_2^n(ab^{n_1}c, ab^{n_2}c, \ldots, ab^{n_{3m}}c)$.
We do induction on $m$. If $m = 1$ the statement is trivial.
Otherwise, write $w = w_1 \cdots w_m$ with $w_i \in L$ having precisely three $a$'s, three $c$'s and
$B$ many times the letter $b$. First, we show an
auxiliary result.
 For a tuple $(x_1, \ldots, x_n)$ and $I \subseteq \{1,\ldots,n\}$
 denote by $(x_1, \ldots, x_n) \uparrow I$ the tuple resulting
 when the entries in $I$ are left out. For example, $(x_1, \ldots, x_n) \uparrow \{ 3,\ldots, n\} = (x_1, x_2)$.
 
\medskip

\noindent\underline{Claim:} Let $u \in L$ and $k \ge 3$.
 Then, $uv \in \shuffle_2^n(ab^{n_1}c, ab^{n_2}c, \ldots, ab^{n_{k}}c)$ 
 if and only if there exist $i_1 < i_2 < i_3$
 such that $v \in \shuffle_2^{n-3}(ab^{m_1}c, \ldots, ab^{m_{k-3}}c)$
 and 
 \[
  (ab^{n_1}c, ab^{n_2}c, \ldots, ab^{n_{k}}c) \uparrow \{ i_1, i_2, i_3 \}
   = (ab^{m_1}c, \ldots, ab^{m_{k-3}}c)
 \]
 and $n_{i_1} + n_{i_2} + n_{i_3} = B$.
 \begin{quote}
     \emph{Proof of the Claim.}
     As $u \in a\{a,b,c\}^*c$, and $u$ contains precisely three $a$'s and $c$'s,
     it must consist of precisely three words from the argument
     and cannot interlace with any other.
     So,
     \[ 
     u \in \shuffle_2^3(ab^{n_{i_1}}c, ab^{n_{i_2}}c, ab^{n_{i_3}}c)
     \mbox{ and } 
     v \in \shuffle_2^{n-3}(ab^{m_1}c, \ldots, ab^{m_{k-3}}c)
     \]
     as stated in the Claim.
     
     Conversely, by the definition
     of $L$, we have $u \in \shuffle_2^3(ab^{r_1}c, ab^{r_2}c, ab^{r_3}c)$
     with $r_1 + r_2 + r_3 = B$.
     Then,
     \begin{multline*}
        \shuffle_2^3(ab^{r_1}c, ab^{r_2}c, ab^{r_3}c) \cdot 
      \shuffle_2^{n-3}(ab^{m_1}c, \ldots, ab^{m_{k-3}}c) \\
      \subseteq 
       \shuffle_2^n(ab^{n_1}c, ab^{n_2}c, \ldots, ab^{n_{k}}c)
     \end{multline*}
     where $r_1, r_2, r_3$ are interlaced into the $m_1,\ldots, m_{k-3}$
     to give the numbering $n_1, \ldots, n_k$.
     This gives the claim. \emph{[End, Proof of the Claim.]}
 \end{quote}
 
 Then, applying the above claim to
 $w = w_1 \cdots w_m$ with $w_i \in L$
 gives \[
  w_2 \cdots w_m \in \shuffle_2^{n-3}(ab^{s_1}c, ab^{s_2}c, \ldots, ab^{s_{3m-3}}c).
  \]
  By induction hypothesis, the remaining numbers $\{s_1, \ldots, s_3\}$
  have a $3$-partition, so the original sequence admits such a partitioning.\qed
\end{proof}


Lastly, as the constructions in the proof of Lemma~\ref{lem:intreg}
are all effective, and the inclusion problem for regular languages
is decidable~\cite{HopUll79}, we can decide if a given regular language is preserved
under any of the shuffle variants.

As the inclusion problem is undecidable even for context-free languages~\cite{HopUll79},
we cannot derive an analogous result for the other families of languages in the same way.

\begin{propositionrep} 
 For every regular language $L \subseteq \Sigma^*$ and $i \in \{1,2\}$, we
 can decide whether $L$ is closed under $\shuffle_i^n$, i.e,
 if $\shuffle_i^n(L,\ldots, L) \subseteq L$ holds.
\end{propositionrep}
\begin{proof}
 For two regular language, we can effectively construct an automaton 
 accepting their shuffle. Hence, we can effectively
 construct an automaton for $L \shuffle \ldots \shuffle L$ ($n$ times).
 As an automaton for the intersection with another regular language
 is effectively constructible, by the method of proof of Lemma~\ref{lem:intreg},
 we find that an automaton for $\shuffle_i^n(L, \ldots, L)$
 is effectively constructible. Because the inclusion problem
 for regular languages is decidable, we 
 can decide if $\shuffle_i^n(L, \ldots, L) \subseteq L$
 holds true.~\qed
\end{proof}

\section{Permuting Arguments}
\label{sec:comm_variants}

If we permute the arguments of the first shuffle variant $\shuffle_1^n$ we may get different results.
Also, for the second variant, see Proposition~\ref{prop:properties_shuffles:cyclic_permutation},
only permuting the arguments cyclically does not change the result, but permuting the arguments arbitrarily might change
the result, see Example~\ref{ex:shuffle_2_abbc}.

Here, we introduce two variants of $\shuffle_1$ that are indifferent
to the order of the arguments, i.e., permuting the arguments does not change
the result, by considering all possibilities
in which the strings could be interlaced. A similar definition is possible for the second variant.

\begin{definition}[$n$-ary symmetric initial literal shuffle] 
\label{def:shuffle_3}
 Let $u_1, \ldots, u_n \in \Sigma^*$ and $x_1, \ldots, x_n \in \Sigma$.
 Then the function $\shuffle_3^n : (\Sigma^*)^n \to \mathcal P(\Sigma^*)$
 is given by
 \[
 \shuffle_3^n(u_1, \ldots, u_n) = \bigcup_{\pi \in \mathcal S_n} \{ \shuffle_1^n(u_{\pi(1)},\ldots,u_{\pi(n)}) \}.
 \]
\end{definition}

An even stronger form as the previous definition do we get, if we do not care
in what order we put the letters at each step.

\todo{nur mit der operation kann man eigentlich sachen ausdrücken die sonst kompliziert sind,
da die anderen auch über shuffle über getrennten alphabeten und schnitt mit regulärer sprache.}
\begin{definition}[$n$-ary non-ordered initial literal shuffle]
\label{def:shuffle_4}
 Let $u_1, \ldots, u_n \in \Sigma^*$ and $x_1, \ldots, x_n \in \Sigma$.
 Then define
 \begin{align*}
  \shuffle_4^n(x_1u_1, \ldots, x_n u_n) & = \bigcup_{\pi \in \mathcal S_n} x_{\pi(1)} \cdots x_{\pi(n)} \shuffle_4^n(u_1, \ldots, u_n)  \\ 
  \shuffle_4^n(u_1, \ldots, u_{j-1}, \varepsilon, u_{j+1},\ldots, u_n) & = \shuffle_4^{n-1}(u_1, \ldots, u_{j-1}, u_{j+1},\ldots, u_n) \\
  \shuffle_4^{1}(u_1) & = \{ u_1 \}.    
 \end{align*} 
\end{definition}

Similarly as in Definition~\ref{def:iterated_versions}
we can define iterated versions $\itshuffle{L}{3}$
and $\itshuffle{L}{4}$ for $L \subseteq \Sigma^*$.
The following properties follow readily.

\begin{propositionrep}
\label{prop:properties_shuffles_comm}
 Let $L_1, \ldots, L_n \subseteq \Sigma^*$, $\pi \in \mathcal S_n$ and $i \in \{3,4\}$. Then 
 \begin{enumerate}
 \item $
  \shuffle_3^n(L_1, \ldots, L_n) = \shuffle_3^n(L_{\pi(1)}, \ldots, L_{\pi(n)});
 $
 \item $
  \shuffle_4^n(L_1, \ldots, L_n) = \shuffle_4^n(L_{\pi(1)}, \ldots, L_{\pi(n)});
 $
 \item \label{prop:properties_shuffles:shuffle_1_in_shuffle_3_in_shuffle_4}
  $\{ \shuffle_1^n(L_1, \ldots, L_n) \} \subseteq \shuffle_3^n(L_1, \ldots, L_n) \} \subseteq \shuffle_4^n(L_1, \ldots, L_n)$;
 \item 
  \label{prop:properties_shuffles:shuffle_i_in_shuffle}
  $\shuffle_i^n(L_1, \ldots, L_n) \subseteq L_1 \shuffle \cdots \shuffle L_n$;
 \item $\Sigma^* = \itshuffle{\Sigma}{i}$ for $i \in \{3,4\}$;
 \item $\itshuffle{L_1}{1}\subseteq \itshuffle{L_1}{3}\subseteq\itshuffle{L_1}{4}\subseteq L_1^{\shuffle,*}$;
 \item for $u_1, \ldots, u_n, u \in \Sigma^*$, 
  if $u \in \shuffle_i^n(\{u_1\}, \ldots, \{u_n\})$, then $|u| = |u_1| + \cdots + |u_n|$.
 \end{enumerate}
\end{propositionrep}
\begin{proof}
\begin{enumerate}
\item Let $L_1, \ldots, L_n \subseteq \Sigma^*$
and $\pi \in\mathcal S_n$ a permutation.
Then
\begin{align*}
 \shuffle_3^n(L_1, \ldots, L_n) & = \bigcup_{u_1 \in L_1, \ldots, u_n \in L_n} \shuffle_3^n(u_1, \ldots u_n) \\ 
        & = \bigcup_{u_1 \in L_1, \ldots, u_n \in L_n} \{ \shuffle_1^n(u_{\tau(1)}, \ldots, u_{\tau(n)}) \mid \tau \in \mathcal S_n \} \\
        & = \bigcup_{u_1 \in L_1, \ldots, u_n \in L_n} \{ \shuffle_1^n(u_{\pi(\tau(1))}, \ldots, u_{\pi(\tau(n))}) \mid \tau \in \mathcal S_n \} \\
        & = \bigcup_{u_1 \in L_1, \ldots, u_n \in L_n} \shuffle_3^n(u_{\pi(1)}, \ldots, u_{\pi(n)}) \\
        & = \shuffle_3^n(L_{\pi(1)}, \ldots, L_{\pi(n)}.
\end{align*}

\item Let $u_1, \ldots, u_n \in \Sigma^*$ and $\pi \in \mathcal S_n$.
 We have $u \in \shuffle_4^n(u_1, \ldots, u_n)$
 if and only if $u$ could be obtained from a  matrix with entries
 in $\Sigma \cup \{\varepsilon\}$ (similarly as in the proof of part (2) in Proposition~\ref{prop:properties_shuffles:cyclic_permutation}, which is constructed by writing the words
 $u_1, \ldots, u_n$ in that order in the rows, padding with $\varepsilon$
 at the end such that every row has the same length, then permuting
 the entries in each column and finally reading these columns and concatenating
 the results in that order. As we permute the entries in each column
 before, it actually does not matter in what order we put the words $u_1, \ldots, u_n$
 in the rows in the first step, i.e., we could start with any permutation
 and get the same set of words at the end. Hence
 $\shuffle_4^n(u_1, \ldots, u_n) = \shuffle_4^n(u_{\pi(1)}, \ldots, u_{\pi(n)})$
 and this is true by extension to languages.
 
\item  Let $u_1, \ldots, u_n \in \Sigma^*$. 
 Then $\shuffle_1^n(u_1, \ldots, u_n) \in \shuffle_3^n(u_1, \ldots, u_n)$
 is obvious. Also, it could be easily shown, using induction
 over $|u_1| + \ldots + |u_n|$ and Definition~\ref{def:shuffle_4}, that
 $\shuffle_1^n(u_1, \ldots, u_n) \in \shuffle_4^n(u_1, \ldots, u_n)$.
 By the previous arguments,  $\shuffle_4$ is invariant under permutation of its arguments.
 Hence
 $$
  \shuffle_3(u_1, \ldots, u_n) = \{ \shuffle_1(u_{\tau(1)},\ldots, u_{\tau(n)}) \mid \tau \in \mathcal S_n \} 
   \subseteq \shuffle_4^n(u_1, \ldots, u_n).
 $$
 This also gives the result for the extension to languages.

\item Suppose $u_1, \ldots, u_n \in \Sigma^*$. 
 As the shuffle operation is commutative, we find 
 $$
  \shuffle_3(u_1, \ldots, u_n) = \{ \shuffle_1(u_{\tau(1)},\ldots, u_{\tau(n)}) \mid \tau \in \mathcal S_n \} 
   \subseteq u_1 \shuffle \ldots \shuffle u_n.
 $$
 This yields, together with Lemma~\ref{lem:shuffle_hom_inclusion},
 \begin{align*} 
  \{ h(\shuffle_1^n(\$^{k_1} u_1, \ldots, \$^{k_n} u_n)) \} 
  & \subseteq h(\$^{k_1} u_1 \shuffle \ldots \shuffle \$^{k_n} u_n) \\
  & \subseteq h(\$^{k_1} u_1) \shuffle \ldots \shuffle h(\$^{k_n} u_n) \\ 
  & = u_1 \shuffle \ldots \shuffle u_n.
 \end{align*}
 So, $\shuffle_2^n(u_1, \ldots, u_n) \subseteq u_1 \shuffle \ldots \shuffle u_n$.
 Hence, these equations are also true by their extensions
 to languages. 
 
\item With Proposition~\ref{prop:properties_shuffles},
 we have $\Sigma^* = \itshuffle{\Sigma}{1}$.
Then, with Property~\ref{prop:properties_shuffles:shuffle_1_in_shuffle_3_in_shuffle_4}, $\itshuffle{\Sigma}{1} \subseteq \itshuffle{\Sigma}{3} \subseteq \itshuffle{\Sigma}{4}$,
which gives the claim.

\item The first inclusion is implied by Property~\ref{prop:properties_shuffles:shuffle_1_in_shuffle_2}
 of this Proposition~\ref{prop:properties_shuffles}.
 The other inclusions are implied by Property~\ref{prop:properties_shuffles:shuffle_1_in_shuffle_3_in_shuffle_4} and Property~\ref{prop:properties_shuffles:shuffle_i_in_shuffle}.
 
\item By Property~\ref{prop:properties_shuffles:shuffle_i_in_shuffle}, the shuffle variants
 are included in the general shuffle. Then, argue as 
 in the proof of the corresponding statement in Proposition~\ref{prop:properties_shuffles}.
\end{enumerate}
\end{proof}

With these properties, we find that for the iteration the first and third
shuffle variants give the same language operator.

\begin{lemmarep} \label{lem:shuffle1_and_3_equal_for_languages}
 For languages $L \subseteq \Sigma^*$
 we have $\shuffle_1^n(L, \ldots, L) = \shuffle_3^n(L, \ldots, L)$.
\end{lemmarep}
\begin{proof}
 By Proposition~\ref{prop:properties_shuffles}, $\shuffle_1^n(L, \ldots, L) \subseteq \shuffle_3^n(L, \ldots, L)$.
 Conversely, suppose $u \in \shuffle_3^n(u_1, \ldots, u_n)$
 with $u_1, \ldots, u_n \in L$.
 Then $u = \shuffle_1^n(u_{\pi(1)}, \ldots, u_{\pi(n)})$
 for some $\pi \in \mathcal S_n$.
 But $\{ u_1, \ldots, u_n \}\subseteq L$, so we get $u \in \shuffle_1^n(L, \ldots, L)$. 
\end{proof}

For the iterated version, this gives that the first and third variant are equal.

\begin{corollary}\label{cor:itshuffle_1_3_equal}
  Let $L \subseteq \Sigma^*$
be a language. Then 
 $\itshuffle{L}{1} = \itshuffle{L}{3}$.
\end{corollary}

Hence, with Proposition~\ref{prop:iterated_variant_non_cf},
we can deduce that $\itshuffle{(abc)}{3} = \itshuffle{(abc)}{4}\cap a^*b^*c^* = \{ a^m b^m c^m \mid m \ge 0\}$
and so even for finite languages, the iterated shuffles
yield languages that are not context-free.

We find that Lemma~\ref{lem:intreg}, Proposition~\ref{prop:trio_shuffles_closures}, Proposition~\ref{prop:hom_shuffles}
and Proposition~\ref{prop:chomsky_hierarchy_closures} also hold for the third and fourth shuffle variant.
To summarize:

\begin{propositionrep} Let $i \in \{3,4\}$. Then:
\begin{enumerate}
    \item  If $\mathcal L$ is a family of languages closed under intersection
     with regular languages, isomorphic mappings and (general) shuffle, then $\mathcal L$
     is closed under $\shuffle_i^n$ for $i \in \{3,4\}$ and each $n \ge 1$.
    \item If $\mathcal L$ is a full trio, then $\mathcal L$
    is closed under shuffle if and only if it is closed under~$\shuffle_i^n$.
    \item For regular $L \subseteq \Sigma^*$, it is decidable
    if $\shuffle_i^n(L, \ldots, L) \subseteq L$.
    
    \item For $L \subseteq \Sigma^*$ and the homomorphism $h : (\Sigma \cup \$)^* \to \Sigma^*$
    given by $h(x) = x$ for $x\in \Sigma$ and $h(\$) = \varepsilon$,
    we have $L^{\shuffle,*} = h(\itshuffle{h^{-1}(L)}{i})$.
    
    \item The families of regular, context-sensitive, recursive and recursively
    enumerable languages are closed under $\shuffle_i^n$ and the families of context-sensitive, recursive and recursively enumerable
    languages are closed for $\itshuffle{L}{i}$.
\end{enumerate}
\end{propositionrep}
\begin{proof}
\begin{enumerate}
    \item   Suppose $\mathcal L$ is closed under intersection with regular
 languages and shuffle.
 The method of proof is similar to the one used in the proof of Lemma~\ref{lem:intreg}.
 Again, for the sake of notational simplicity, we shall give the proof, which
 is perfectly general, in the special case $n = 3$.
 Let $L_1, L_2, L_3 \in \mathcal L$ be three languages over the disjoint
 alphabets $\Sigma_1, \Sigma_2, \Sigma_3$.
 For a permutation $\pi \in \mathcal S_3$ set
 $$
  S^{\pi} = ( (\Sigma_{\pi(1)} \Sigma_{\pi(2)})^* \cup (\Sigma_{\pi(1)} \Sigma_{\pi(3)})^* \cup (\Sigma_{\pi(2)} \Sigma_{\pi(3)})^* )
      (\Sigma_1^* \cup \Sigma_2^* \cup \Sigma_3^*)
 $$
 with the shorthand notation $S = S^{\operatorname{id}}$. Also, set
 $$
  T =  ( (\Sigma_1 \Sigma_2 \cup \Sigma_2 \Sigma_1)^* \cup 
         (\Sigma_1 \Sigma_3 \cup \Sigma_3 \Sigma_1)^* \cup 
         (\Sigma_2 \Sigma_3 \cup \Sigma_3 \Sigma_2)^* )
       (\Sigma_1^* \cup \Sigma_2^* \cup \Sigma_3^*)
 $$
 and 
 \begin{align*}
     R_3 & = \bigcup_{\pi \in \mathcal S_3} (\Sigma_{\pi(1)} \Sigma_{\pi(2)} \Sigma_{\pi(3)})^* S^{\pi}, \\ 
     R_4 & = \left( \bigcup_{\pi \in \mathcal S_3} \Sigma_{\pi(1)} \Sigma_{\pi(2)} \Sigma_{\pi(3)}) \right)^* T.
 \end{align*}
 Then
 \begin{align*}
     \shuffle_3^3(L_1, L_2, L_3) & = (L_1 \shuffle L_2 \shuffle L_3) \cap R_3, \\
     \shuffle_1^4(L_1, L_2, L_3) & = (L_1 \shuffle L_2 \shuffle L_3) \cap R_4.
 \end{align*}
 Hence, the can write all our shuffle variants with the general shuffle and intersection
 by a suitable regular language. 
 
\item  First suppose $\mathcal L$ is a full trio that is closed under shuffle.
 Then, by the previous statement it is closed
 under each operation $\shuffle_n^i$ for $i \in \{3,4\}$.
 Conversely, suppose $\mathcal L$ is closed under $\shuffle_i^n$
 for $i \in \{3,4\}$ with $n \ge 2$.
 We distinguish several cases. Note~\cite{DBLP:journals/actaC/Berard87,DBLP:journals/tcs/Berard87} 
     that, for the homomorphism $h : (\Sigma \cup \{\$\})^* \to \Sigma^*$
     given by $h(x) = x$ for $x \in \Sigma$ and $h(\$) = \varepsilon$ with $\$ \notin \Sigma$,
     we have
     \begin{equation}\label{eqn:init_shuffle_hom_shuffle2}
      L_1 \shuffle L_2 = h(\shuffle_1^2(h^{-1}(L_1), h^{-1}(L_2))). 
     \end{equation}
 \begin{enumerate}
 \item[(i)] $i = 3$.
    
     Let $L_1, L_2 \in \mathcal L$. We have $\shuffle_3^2(L_1, L_2) = \shuffle_1^2(L_1, L_2) \cup \shuffle_1^2(L_2, L_1)$.
     Hence, by Equation~\eqref{eqn:init_shuffle_hom_shuffle2},
     \begin{align*} 
      & h(\shuffle_3^2(h^{-1}(L_1), h^{-1}(L_2))) \\
       & = h(\shuffle_1^2(h^{-1}(L_1), h^{-1}(L_2))) \cup h(\shuffle_3^2(h^{-1}(L_2), h^{-1}(L_1))) \\ 
       & = L_1 \shuffle L_2 \cup L_2 \shuffle L_1 \\
       & = L_1 \shuffle L_2.
     \end{align*}
     This show the claim for $n = 2$.
     For $n > 2$, as $\shuffle_3^2(L_1, L_2, \{\varepsilon\}, \ldots, \{\varepsilon\}) = \shuffle_3^2(L_1, L_2)$,
     the claim is implied as it is valid for $n = 2$.
     
 \item[(ii)] $i = 4$.
 
   We have $\shuffle_1^2(h^{-1}(L_1), h^{-1}(L_2)) \subseteq \shuffle_4^2(h^{-1}(L_1), h^{-1}(L_2))$.
   So, by Equation~\eqref{eqn:init_shuffle_hom_shuffle2}, Proposition~\ref{prop:properties_shuffles_comm} and Lemma~\ref{lem:shuffle_hom_inclusion},
   \begin{align*}
       L_1 \shuffle L_2 & \subseteq h(\shuffle_4^2(h^{-1}(L_1), h^{-1}(L_2)) \\ 
                        & \subseteq h(h^{-1}(L_1) \shuffle h^{-1}(L_2)) \\
                        & \subseteq h(h^{-1}(L_1) \shuffle h(h^{-1}(L_2)) \\
                        & \subseteq L_1 \shuffle L_2.
   \end{align*}

 \end{enumerate}
 So, in each case we can represent the shuffle, which gives our claim.
 
\item  For two regular language, we can effectively construct an automaton 
 accepting their shuffle. Hence, we can effectively
 construct an automaton for $L \shuffle \ldots \shuffle L$ ($n$ times).
 As an automaton for the intersection with another regular language
 is effectively constructible, by the method of proof of part (1),
 we find that an automaton for $\shuffle_i^n(L, \ldots, L)$
 is effectively constructible. Because the inclusion problem
 for regular languages is decidable, we 
 can decide if $\shuffle_i^n(L, \ldots, L) \subseteq L$
 holds true. 
 
\item  By Proposition~\ref{prop:properties_shuffles_comm},
 for $i \in \{3, 4\}$,
 $\itshuffle{h^{-1}(L)}{i} \subseteq (h^{-1}(L))^{\shuffle,*}$.
 Hence, using Lemma~\ref{lem:shuffle_hom_inclusion},
 we can conclude $h(\itshuffle{h^{-1}(L)}{i}) \subseteq h((h^{-1}(L))^{\shuffle,*})
 \subseteq (h(h^{-1}(L)))^{\shuffle,*} \subseteq L^{\shuffle,*}$.
 Similarly, by Proposition~\ref{prop:properties_shuffles_comm},
 we find that for $i \in \{2,3,4\}$
 we have
 $$
  \itshuffle{h^{-1}(L)}{1} \subseteq  \itshuffle{h^{-1}(L)}{i} \subseteq h^{-1}(L)^{\shuffle,*}
 $$
 and so, with Lemma~\ref{lem:shuffle_hom_inclusion},
 $$ 
  h(\itshuffle{h^{-1}(L)}{1}) \subseteq  h(\itshuffle{h^{-1}(L)}{i}) \subseteq L^{\shuffle,*}.
 $$
 By Proposition~\ref{prop:hom_shuffles},
 we have $L^{\shuffle,*} \subseteq h(\itshuffle{h^{-1}(L)}{1})$ and the claim follows.
 
\item We refer to the construction done in the proof of Proposition~\ref{prop:chomsky_hierarchy_closures}
 and only indicate were it has to be modified. The argument for closure of each class
 under $\shuffle_3^n$ and $\shuffle_4^n$ is the same.
 
 \medskip 

\noindent\underline{The case $i = 3$}:  By Lemma~\ref{lem:shuffle1_and_3_equal_for_languages} this reduces to the case $i = 1$.
 
\medskip 

\noindent\underline{The case $i = 4$}: Let $u$ be the input.
 First, non-deterministically, guess a number of operands $1 \le n \le |u|$.
 Then, construct a sequence by setting $n_1 = n$
 and, if $n_i$ was choosen and $|u| - (n_1 + \ldots + n_i) > 0$,
 guess $1 \le n_{i+1} \le \min\{ |u| - (n_1 + \ldots + n_i), n_i \}$. 
 This sequence corresponds to a structuring of the input $u$
 in the first $n_1$ symbols, then the next $n_2$ symbols and so on.
 This sequence could be stored by marking the input, for example by copying
 it on an additional working tape and set marking symbols.
 Then, inside each marked factor of $u$, apply some permutation to it.
 Afterward, erase all markings.  
 Then, operate as for the case $i = 1$, but use the altered string as
 input. Similar to the case $i = 1$, we see that all closure properties are
 satisfied.

\medskip 

These constructions, together with those from
the proof of Proposition~\ref{prop:chomsky_hierarchy_closures},
show closure under all iterated shuffle variants.  
We have used closure properties to show closure under $\shuffle_i^n$
for $i \in \{1,\ldots,4\}$ above. But note that we could also use the referred constructions,
but with $n$ fixed instead of a guessed value.  
\end{enumerate} 
\end{proof}

Lastly, we give two examples.

\begin{example}
\label{ex:iterated_shuffles}
 Set $L = \{ab,ba\}$. Define the homomorphism $g : \{a,b\}^* \to \{a,b\}^*$
 by $g(a) = b$ and $g(b) = a$, i.e., interchanging the symbols.
 \begin{enumerate}
 \item $\itshuffle{L}{1} = \itshuffle{L}{3} = \{ u g(u)  \mid u \in \{a,b\}^* \}$.
  %
  %
  \begin{proof}
   Let $u_1, \ldots, u_n \in L$ and $w = \shuffle_1^n(u_1, \ldots, u_n)$. Then
   $w = I(u_1, \ldots, u_n)$. As $|u_1| = \ldots = |u_n| = 2$,
   we can write $w = x_1 \cdots x_{2n}$
   with $x_i \in \{a,b\}$ for $i \in \{1,\ldots, 2n\}$.
   By Definition~\ref{def:interleaving_operator}
   of the $I$-operator, we have, for $i \in \{1,\ldots, n\}$,
   $u_i = x_i x_{i+n}$.
   So, if $u_i = ab$, then $x_i = a$ and $x_{i+n} = b$,
   and if $u_i = ba$, then $x_i = b$ and $x_{i+n} = a$.
   By Corollary~\ref{cor:itshuffle_1_3_equal}, $\itshuffle{L}{1} = \itshuffle{L}{3}$.
   \end{proof}

 \item $\itshuffle{L}{4} = \{ uv \mid u,v \in \{a,b\}^*, |u|_a = |v|_b, |u|_b = |v|_a \}$.
 
  \begin{proof} Let $u_1, \ldots, u_n \in L$ and $w = \shuffle_4^n(u_1, \ldots, u_n)$.
   Then, using the inductive Definition~\ref{def:shuffle_4} twice, we find $w = uv$,
   where $u$ contains all the first symbols of the arguments $u_1, \ldots, u_n$ in some order,
   and $v$ all the second symbols. Hence, for each $a$ in $u$, we must have a $b$ in $v$
   and vice versa. So, $w$ is contained in the set on the right hand side.
   Conversely, suppose $w = uv$ with $|u|_a = |v|_b$ and $|u|_b = |v|_a$.
   Then set $n = |u| = |v|$, $u = x_1 \cdots x_n$, $v = y_1 \cdots y_n$
   with $x_i, y_i \in \Sigma$ for $i \in \{1,\ldots, n\}$. We can reorder
   the letters to match up, i.e., an $a$ with a $b$ and vice versa. More specifically,
   we find a permutation $\pi \in \mathcal S_n$
   such that
   $w = I(x_{\pi(1)} y_1, \ldots, x_{\pi(n)} y_n) \in \shuffle_4^n(x_1 y_1, \ldots, x_n y_n)
   \in \itshuffle{L}{4}$. 
  \end{proof}
 \end{enumerate}
\end{example}

\section{Conclusion and Summary} 
 
The literal and the initial literal shuffle were introduced with the idea to describe
the execution histories of step-wise synchronized processes.
However, a closer investigation revealed that they only achieve this for
two processes, and, mathematically, the lack of associativity of these
operations prevented usage for more than two processes.
Also, iterated variants derived from this non-associative binary operations
depend on a fixed bracketing and does not reflect the synchronization of $n$
processes.
The author of the original papers~\cite{DBLP:journals/tcs/Berard87,DBLP:journals/actaC/Berard87} 
does not discuss this issue, but is more concerned
with the formal properties themselves.
Here, we have introduced two operations that lift
the binary variant to an arbitrary number of arguments, hence allowing
simultaneous step-wise synchronization of an arbitrary number of processes.
We have also introduced iterative variants, which are more natural than
the previous ones for the non-associative binary operations.
In summary, we have
\begin{enumerate}
    \item investigated the formal properties and relations between our shuffle variant
     operations,
    \item we have found out that some properties are preserved in analogy to the binary
     case, but others are not, for example commutativity, see Proposition~\ref{prop:properties_shuffles};
    \item we have shown various closure or non-closure properties for the family of languages
     from the Chomsky hierarchy and the recursive languages;
    \item used one shuffle variant to characterize the family of finite languages;
    \item in case of a full trio, we have shown that their expressive power coincides
     with the general shuffle, and, by results from~\cite{DBLP:journals/tcs/Berard87,DBLP:journals/actaC/Berard87},
     with the initial literal and literal shuffle;
    \item we have investigated various decision problems, some of them
     are tractable even if an analogous decision problem with the general shuffle operation
     is intractable. However, we have also identified an intractable decision problem
     for our second $n$-ary shuffle variant.
\end{enumerate}

As seen in Proposition~\ref{prop:np_problem} for the \NP-complete decision problem, we needed
an alphabet of size at least three. For an alphabet of size one, i.e., an unary alphabet, 
the $n$-ary shuffle
for any of the variants considered reduces to the ($n$-times) concatenation
of the arguments, which is easily computable. Then, deciding if the result
of this concatenation is contained in a given regular language could be done in polynomial
time.
So, a natural question is if the problem formulated in Proposition~\ref{prop:np_problem}
remains \NP-complete for binary alphabets only. Also, it is unknown 
what happens if we alter the problem by not allowing a finite language represented by a deterministic automaton as input,
but only a single word, i.e., $|L| = 1$. For the general shuffle, this problem
is \NP-complete, but it is open what happens if we use the second $n$-ary variant.
Also, it is unknown if the problem remains \NP-complete if we represent the input not
by an automaton, but by a finite list of words.

\medskip \noindent {\textbf{Acknowledgement.} I thank  anonymous reviewers
of a previous version for feedback and remarks that helped to improve the presentation. I also thank the reviewers of the current version
for careful reading and pointing out typos and some unclear formulations. 
Due to the strict page limit, I cannot put all proofs into the paper. I have supplied proof sketches where possible. Also, an extended version with all the missing proofs is in preparation.
}



\bibliographystyle{psc}
\bibliography{ms} 

\begin{thebibliography}{B{\'{e}}r87b}

\bibitem[B{\'{e}}r87a]{DBLP:journals/actaC/Berard87}
B{\'{e}}atrice B{\'{e}}rard.
\newblock Formal properties of literal shuffle.
\newblock {\em Acta Cyb.}, 8(1):27--39, 1987.

\bibitem[B{\'{e}}r87b]{DBLP:journals/tcs/Berard87}
B{\'{e}}atrice B{\'{e}}rard.
\newblock Literal shuffle.
\newblock {\em Theor. Comput. Sci.}, 51:281--299, 1987.

\bibitem[GJ79]{garey1979computers}
M.~R. Garey and D.~S. Johnson.
\newblock {\em Computers and Intractability: A Guide to the Theory of
  NP-Completeness (Series of Books in the Mathematical Sciences)}.
\newblock W. H. Freeman, first edition edition, 1979.

\bibitem[HU79]{HopUll79}
John~E. Hopcroft and Jeff~D. Ullman.
\newblock {\em Introduction to Automata Theory, Languages, and Computation}.
\newblock Addison-Wesley Publishing Company, 1979.

\bibitem[WH84]{DBLP:journals/jcss/WarmuthH84}
Manfred~K. Warmuth and David Haussler.
\newblock On the complexity of iterated shuffle.
\newblock {\em J. Comput. Syst. Sci.}, 28(3):345--358, 1984.

\end{thebibliography}


\begin{thebibliography}{10}

\bibitem{DBLP:journals/acta/ArakiT81}
{\sc T.~Araki and N.~Tokura}:
\newblock {\em Flow languages equal recursively enumerable languages}.
\newblock Acta Informatica, 15 1981, pp.~209--217.

\bibitem{DBLP:journals/actaC/Berard87}
{\sc B.~B{\'{e}}rard}:
\newblock {\em Formal properties of literal shuffle}.
\newblock Acta Cyb., 8(1) 1987, pp.~27--39.

\bibitem{DBLP:journals/tcs/Berard87}
{\sc B.~B{\'{e}}rard}:
\newblock {\em Literal shuffle}.
\newblock Theor. Comput. Sci., 51 1987, pp.~281--299.

\bibitem{DBLP:journals/tcs/BerglundBB13}
{\sc M.~Berglund, H.~Bj{\"{o}}rklund, and J.~Bj{\"{o}}rklund}:
\newblock {\em Shuffled languages - representation and recognition}.
\newblock Theor. Comput. Sci., 489-490 2013, pp.~1--20.

\bibitem{BrzozowskiJLRS16}
{\sc J.~A. Brzozowski, G.~Jir{\'{a}}skov{\'{a}}, B.~Liu, A.~Rajasekaran, and
  M.~Szyku{\l}a}:
\newblock {\em On the state complexity of the shuffle of regular languages}, in
  Descrip. Compl. of Formal Systems - 18th {IFIP} {WG} 1.2 International
  Conference, {DCFS} 2016, Bucharest, Romania, July 5-8, 2016. Proceedings,
  C.~C{\^{a}}mpeanu, F.~Manea, and J.~Shallit, eds., vol.~9777 of Lecture Notes
  in Computer Science, Springer, 2016, pp.~73--86.

\bibitem{DBLP:journals/jcss/BussS14}
{\sc S.~Buss and M.~Soltys}:
\newblock {\em Unshuffling a square is {NP}-hard}.
\newblock J. Comput. Syst. Sci., 80(4) 2014, pp.~766--776.

\bibitem{CamHab74}
{\sc R.~H. Campbell and A.~N. Habermann}:
\newblock {\em The specification of process synchronization by path
  expressions}, in Operating Systems OS, E.~Gelenbe and C.~Kaiser, eds.,
  vol.~16 of LNCS, Springer, 1974, pp.~89--102.

\bibitem{garey1979computers}
{\sc M.~R. Garey and D.~S. Johnson}:
\newblock {\em Computers and Intractability: A Guide to the Theory of
  NP-Completeness (Series of Books in the Mathematical Sciences)}, W. H.
  Freeman, first edition~ed., 1979.

\bibitem{Ginsburg75}
{\sc S.~Ginsburg}:
\newblock {\em Algebraic and Automata-Theoretic Properties of Formal
  Languages}, Elsevier Science Inc., USA, 1975.

\bibitem{GinsburgGreibach67}
{\sc S.~{Ginsburg} and S.~{Greibach}}:
\newblock {\em Abstract families of languages}, in 8th Annual Symposium on
  Switching and Automata Theory (SWAT 1967), 1967, pp.~128--139.

\bibitem{DBLP:journals/eatcs/HenshallRS12}
{\sc D.~Henshall, N.~Rampersad, and J.~O. Shallit}:
\newblock {\em Shuffling and unshuffling}.
\newblock Bull. {EATCS}, 107 2012, pp.~131--142.

\bibitem{HopUll79}
{\sc J.~E. Hopcroft and J.~D. Ullman}:
\newblock {\em Introduction to Automata Theory, Languages, and Computation},
  Addison-Wesley Publishing Company, 1979.

\bibitem{Iwana83}
{\sc K.~Iwama}:
\newblock {\em Unique decomposability of shuffled strings: A formal treatment
  of asynchronous time-multiplexed communication}, in Proceedings of the
  Fifteenth Annual ACM Symposium on Theory of Computing, STOC '83, New York,
  NY, USA, 1983, Association for Computing Machinery, p.~374–381.

\bibitem{Iwama:1983:UPU}
{\sc K.~Iwama}:
\newblock {\em The universe problem for unrestricted flow languages}.
\newblock Acta Informatica, 19(1) Apr.~1983, pp.~85--96.

\bibitem{DBLP:journals/tcs/Jantzen81}
{\sc M.~Jantzen}:
\newblock {\em The power of synchronizing operations on strings}.
\newblock Theor. Comput. Sci., 14 1981, pp.~127--154.

\bibitem{DBLP:journals/tcs/Jantzen85}
{\sc M.~Jantzen}:
\newblock {\em Extending regular expressions with iterated shuffle}.
\newblock Theor. Comput. Sci., 38 1985, pp.~223--247.

\bibitem{DBLP:journals/ipl/Jedrzejowicz83}
{\sc J.~Jedrzejowicz}:
\newblock {\em On the enlargement of the class of regular languages by the
  shuffle closure}.
\newblock Inf. Process. Lett., 16(2) 1983, pp.~51--54.

\bibitem{DBLP:journals/tcs/JedrzejowiczS01}
{\sc J.~Jedrzejowicz and A.~Szepietowski}:
\newblock {\em Shuffle languages are in {P}}.
\newblock Theor. Comput. Sci., 250(1-2) 2001, pp.~31--53.

\bibitem{DBLP:journals/fuin/KudlekF14}
{\sc M.~Kudlek and N.~E. Flick}:
\newblock {\em Properties of languages with catenation and shuffle}.
\newblock Fundam. Inform., 129(1-2) 2014, pp.~117--132.

\bibitem{DBLP:journals/jcss/Latteux79}
{\sc M.~Latteux}:
\newblock {\em C{\^{o}}nes rationnels commutatifs}.
\newblock J. Comp. Sy. Sc., 18(3) 1979, pp.~307--333.

\bibitem{Nivat82}
{\sc M.~Latteux}:
\newblock {\em Behaviors of processes and synchronized systems of processes},
  in Theoretical Foundations of Programming Methodology, S.~G. Broy~M., ed.,
  vol.~91 of NATO Advanced Study Institutes Series (Series C — Mathematical
  and Physical Sciences), Springer, Dordrecht, 1982, pp.~473--551.

\bibitem{DBLP:journals/tcs/MateescuRS98}
{\sc A.~Mateescu, G.~Rozenberg, and A.~Salomaa}:
\newblock {\em Shuffle on trajectories: Syntactic constraints}.
\newblock Theor. Comput. Sci., 197(1-2) 1998, pp.~1--56.

\bibitem{DBLP:conf/mfcs/Marzurkiewicz75}
{\sc A.~W. Mazurkiewicz}:
\newblock {\em Parallel recursive program schemes}, in Mathematical Foundations
  of Computer Science 1975, 4th Symposium, Mari{\'{a}}nsk{\'{e}} L{\'{a}}zne,
  Czechoslovakia, September 1-5, 1975, Proceedings, J.~Becv{\'{a}}r, ed.,
  vol.~32 of Lecture Notes in Computer Science, Springer, 1975, pp.~75--87.

\bibitem{Nivat68}
{\sc M.~Nivat}:
\newblock {\em Transductions des langages de chomsky}.
\newblock Annales de l'Institut Fourier, 18(1) 1968, pp.~339--455.

\bibitem{Ogden78}
{\sc W.~F. Ogden, W.~E. Riddle, and W.~C. Round}:
\newblock {\em Complexity of expressions allowing concurrency}, in Proceedings
  of the 5th ACM SIGACT-SIGPLAN Symposium on Principles of Programming
  Languages, POPL '78, New York, NY, USA, 1978, Association for Computing
  Machinery, p.~185–194.

\bibitem{DBLP:journals/cl/Riddle79}
{\sc W.~E. Riddle}:
\newblock {\em An approach to software system behavior description}.
\newblock Comput. Lang., 4(1) 1979, pp.~29--47.

\bibitem{DBLP:journals/cl/Riddle79a}
{\sc W.~E. Riddle}:
\newblock {\em An approach to software system modelling and analysis}.
\newblock Comput. Lang., 4(1) 1979, pp.~49--66.

\bibitem{Shaw78zbMATH03592960}
{\sc A.~C. {Shaw}}:
\newblock {\em {Software descriptions with flow expressions.}}
\newblock {IEEE Trans. Softw. Eng.}, 4 1978, pp.~242--254.

\bibitem{stockmeyer1973word}
{\sc L.~J. Stockmeyer and A.~R. Meyer}:
\newblock {\em Word problems requiring exponential time (preliminary report)},
  in Proceedings of the fifth annual ACM Symposium on Theory of Computing,
  STOC, ACM, 1973, pp.~1--9.

\bibitem{DBLP:journals/tcs/BeekMM05}
{\sc M.~H. ter Beek, C.~Mart{\'{\i}}n{-}Vide, and V.~Mitrana}:
\newblock {\em Synchronized shuffles}.
\newblock Theor. Comput. Sci., 341(1-3) 2005, pp.~263--275.

\bibitem{DBLP:journals/jcss/WarmuthH84}
{\sc M.~K. Warmuth and D.~Haussler}:
\newblock {\em On the complexity of iterated shuffle}.
\newblock J. Comput. Syst. Sci., 28(3) 1984, pp.~345--358.

\end{thebibliography}


\end{document}